\documentclass[conference]{IEEEtran}

\usepackage{cite}

\ifCLASSINFOpdf
   \usepackage[pdftex]{graphicx}
   \DeclareGraphicsExtensions{.pdf,.jpeg,.png}
\else
\fi

\usepackage{algorithm}
\usepackage{algorithmic}

  \usepackage{subfigure}

\usepackage{amsmath}
\usepackage{amssymb}
\usepackage{bold-extra}
\usepackage{color}
\usepackage{graphicx}
\usepackage{amsthm}
\usepackage{csquotes}
\usepackage{tikz}
\usepackage{url}
\usepackage[T1]{fontenc}
\usepackage[utf8]{inputenc}
\usetikzlibrary{calc}
\usetikzlibrary{positioning}
\usepackage[absolute,overlay]{textpos} 
\usepackage{latexsym}
\usepackage{csquotes}
\usepackage{color}

\newcommand{\xuadd}{}
\newcommand{\forget}[1]{}

\newtheorem{lemma}{Lemma}
\newtheorem{define}{Definition}
\newtheorem{theorem}{Theorem}

\newtheorem{property}{Property}

\newcommand{\summation}{\sum}

\newcommand{\DSAI}{\textsf{\upshape{D-SAI}}}

\newcommand{\GLI}{\textsf{\upshape{G-LI}}}

\newcommand{\GMEL}{\textsf{\upshape{G-MEL}}}
\newcommand{\FLI}{\textsf{\upshape{F-LI}}} 
 
\newcommand{\DXU}{\textsf{\upshape{D-XU}}} 

\newcommand{\firstSFS}{\textsf{\upshape{SF[x+1]}}} 
\newcommand{\secondSFS}{\textsf{\upshape{SF[x+2]}}} 
\newcommand{\Sched}{\textsf{\upshape{Sched}}} 
\newcommand{\Partition}{\textsf{\upshape{Partition}}} 
\newcommand{\Scrape}{\textsf{\upshape{Scrape}}} 

\newcommand{\platform}{\pi}
\newcommand{\speed}{{{\delta}}}
\newcommand{\processor}{P}
\newcommand{\response}{R}
\newcommand{\sumspeed}{S}
\newcommand{\platfeature}{{{\lambda}}}

\newcommand{\consume}{\chi}

\begin{document}

\IEEEoverridecommandlockouts

\title{Semi-Federated Scheduling of Parallel Real-Time Tasks on Multiprocessors}

\author{\IEEEauthorblockN{Xu Jiang$^{1,2}$, Nan Guan$^1$, Xiang Long$^{2}$, Wang Yi$^3$}
\IEEEauthorblockA{
	~~\\
$^1$The Hong Kong Polytechnic University, Hong Kong\\
$^2$Beihang University, China\\
$^3$Uppsala University, Sweden}}

\maketitle

\begin{abstract}
Federated scheduling is a promising approach to schedule parallel real-time tasks on multi-cores,
where each heavy task exclusively executes on a number of dedicated processors, while light tasks are treated as sequential sporadic tasks and share the remaining processors. 
However, federated scheduling suffers resource waste 
since a heavy task with processing capacity requirement $x + \epsilon$ (where $x$ is an integer and $0 < \epsilon < 1$) 
needs $x + 1$ dedicated processors. In the extreme case, almost half of the processing capacity is wasted. In this paper we propose the
semi-federate scheduling approach, which only grants $x$ dedicated processors to a heavy task with processing capacity requirement $x + \epsilon$,
and schedules the remaining $\epsilon$ part together with light tasks on shared processors. Experiments
with randomly generated task sets show the semi-federated scheduling approach significantly outperforms not only federated scheduling, but also all existing approaches for scheduling parallel real-time tasks on multi-cores.
\end{abstract}
\IEEEpeerreviewmaketitle
\section{Introduction}

Multi-cores are more and more widely used in real-time systems, to meet their rapidly increasing requirements in performance and energy efficiency.
The processing capacity of multi-cores is not a free lunch. Software must be properly parallelized to fully exploit the computation capacity of multi-core processors.
Existing scheduling and analysis techniques for sequential real-time tasks are hard to 
migrate to the parallel workload setting.  
New scheduling and analysis techniques are required to deploy parallel real-time tasks on multi-cores.

A parallel real-time task is usually modeled as a
Directed Acyclic Graph (DAG).
Several scheduling algorithms have been proposed to schedule DAG tasks in recent years, among which \textit{Federated Scheduling}  \cite{li2014analysis} is a promising approach with both
good real-time performance and high flexibility.
In federated scheduling, DAG tasks are classified into \textit{heavy} tasks (density $>1$) and \textit{light} tasks (density $\leq 1$). Each heavy task exclusively executes
on a subset of dedicated processors.
Light tasks are treated as traditional sequential real-time tasks and share the remaining processors. Federated scheduling
not only can schedule a large portion of DAG task systems that is not schedulable by other approaches, but also provides the best quantitative worst-case performance guarantee \cite{li2014analysis}.
On the other hand, federated scheduling allows flexible workload specification as the underlying analysis techniques only require information about the critical path length and total workload of the DAG, and thus can be easily extended to more expressive models, such as DAG with conditional branching \cite{baruah2015global,melani2015response}.

However, federated scheduling may suffer significant resource waste, since each heavy task \textit{exclusively} owns a subset of processors. For example, if a heavy task requires processing capacity $x+\epsilon$ (where $x$ is an integer and $0 < \epsilon < 1$), then $\lceil x+\epsilon \rceil = x+1$ dedicated processors are granted to it, as shown in Figure \ref{fig:fedandsemi}-(a). In the extreme case, almost half of the total processing capacity is wasted
(when a DAG requires $1+\epsilon$ processing capacity and $\epsilon  \rightarrow 0$).

\begin{figure}[!t]
	\centering
	\subfigure[Federated scheduling.]{\includegraphics[width=1.8in]{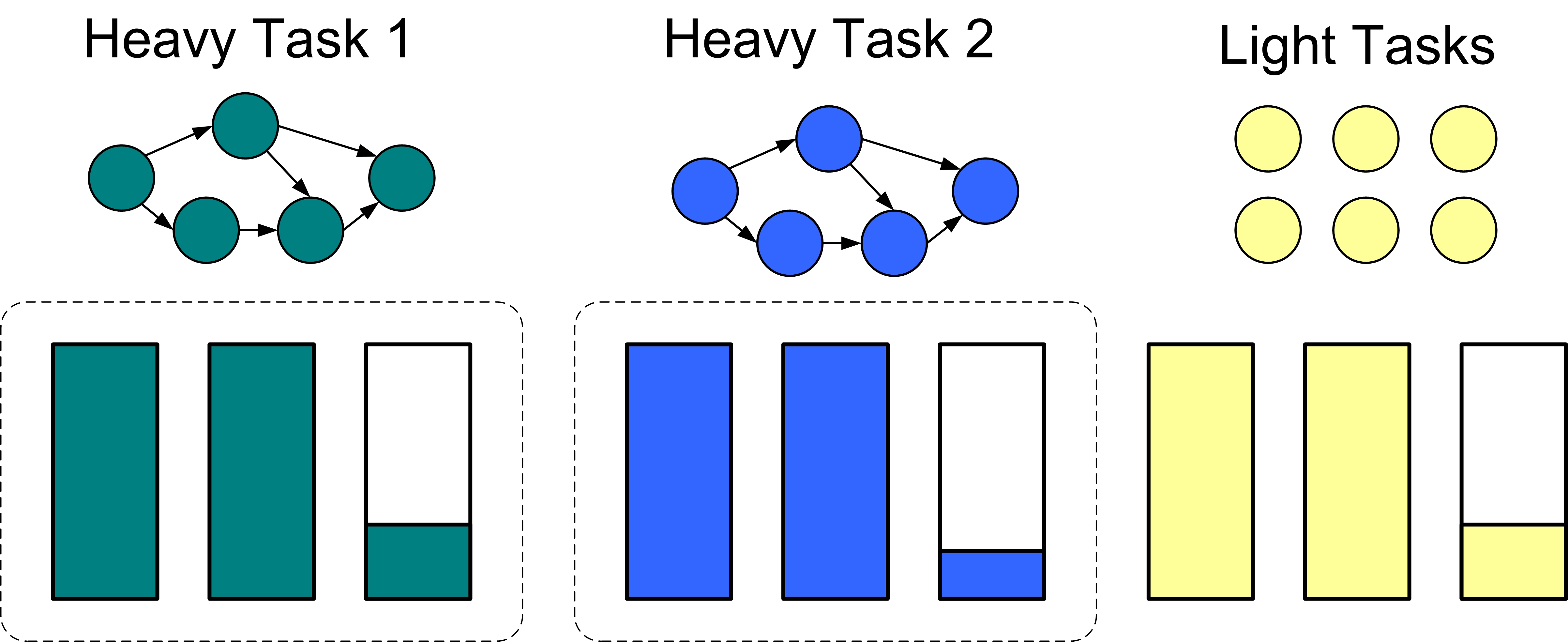}}
	\hspace{0.08in}
	\subfigure[Semi-federated scheduling.]{\includegraphics[width=1.4in]{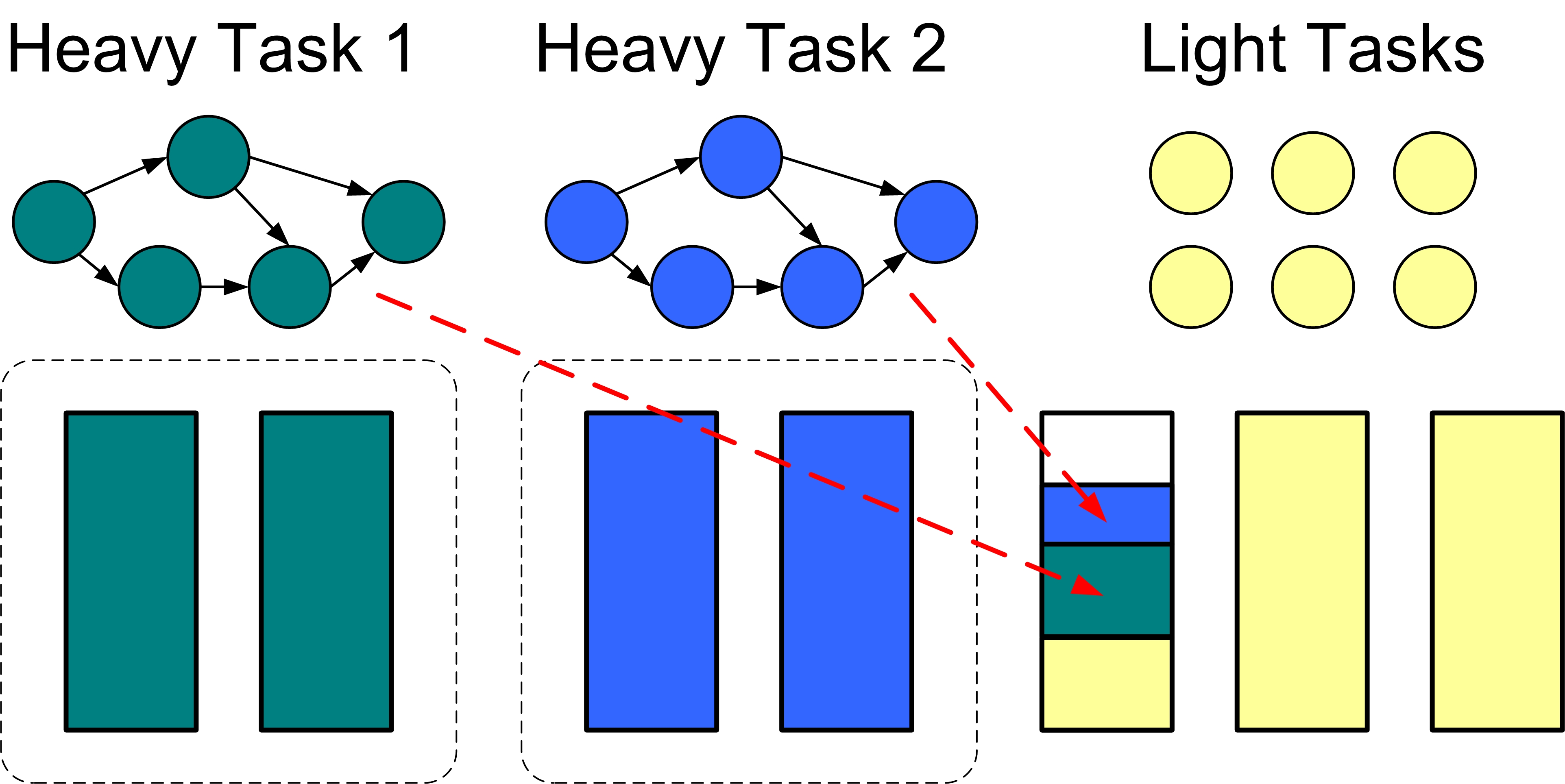}}
	\caption{Illustration of federated scheduling and semi-federated scheduling.} 
	\label{fig:fedandsemi}
\end{figure}

In this work, we propose the \textit{Semi-Federated Scheduling} approach to solve the above resource waste problem. In semi-federated scheduling, a DAG task requiring $x+\epsilon$ processing capacity is only granted $x$ dedicated processors, and the remaining fractional part $\epsilon$ is scheduled together with the light tasks, as illustrated in Figure \ref{fig:fedandsemi}-(b).

The major challenge we face in realizing semi-federated scheduling is how to control and analyze the interference suffered by the fractional part, and its effect to the timing behavior of the entire heavy task. The fractional part of a heavy task is scheduled together with,
and thus suffers interference from
the light tasks and the fractional parts of other heavy tasks. Due to the intra-task dependencies inside a DAG, this interference is propagated to 
other parts of the DAG executed on the dedicated processors, and thus affects the timing behavior of the entire DAG task. 
Existing scheduling and analysis techniques for federated scheduling (based on the classical work in \cite{graham1969bounds}) cannot handle such extra interference.
%

This paper addresses the above challenges and develops semi-federated scheduling algorithms in the following steps.

%
%
%
%
%
First, we study the problem of bounding the response time of an individual DAG executing on a \textit{uniform} multiprocessor platform (where processors have different speeds). The results 
we obtained for this problem
serve as the theoretical foundation of the semi-federated scheduling approach.
Intuitively, we grant a portion ($<1$) of the processing capacity of a processor to execute the fractional part of a DAG, which is similar to executing it on a slower processor.

Second, the above results are transferred to the realistic situation where
the fractional parts of DAG tasks and 
the light tasks share several processors with unit speed. This is realized by executing the fractional parts via sequential \textit{container tasks}, each of which has a \textit{load bound}. A container task plays the role of a dedicated processor with a slower speed (equals the container task's load bound), and thus the above results can be applied to analyze the response time of the DAG task.
	
Finally, we propose two semi-federated scheduling algorithms based on the above framework.
	In the first algorithm, a DAG task requiring $x + \epsilon$ processing capacity is granted $x$ dedicated processors and \textit{one} container task with load bound $\epsilon$, and all the container tasks and the light tasks are scheduled by \textit{partitioned} EDF on the remaining processors. The second algorithm enhances the first one by allowing to divide the fractional part $\epsilon$ into \textit{two} container tasks, which further improves resource utilization.

We conduct experiments with randomly generated workload, which show our semi-federated scheduling algorithms significantly improve schedulability
over the state-of-the-art of, not only federated scheduling, but also the other types such as global scheduling and decomposition-based scheduling.

\section{Preliminary}\label{s:pre}

\subsection{Task Model}\label{ss:task-model}
We consider a task set $\tau$ consisting of $n$ tasks $\{\tau_1,\tau_2,...,\tau_n\}$,
executed on $m$ identical processors with unit speed.
Each task is represented by a DAG, with a period $T_i$ and a relative deadline $D_i$.
We assume all tasks to have \textit{constrained deadlines}, i.e., $D_i \leq T_i$.
Each task is represented by a directed acyclic graph (DAG).
A vertex $v$ in the DAG has a WCET $c(v)$.
Edges represent dependencies among vertices. A directed edge from vertex $v$ to $u$ means that $u$ can only be executed after $v$ is finished.
In this case, $v$ is a \emph{predecessor} of $u$, and $u$ is a \emph{successor} of $v$.
We say a vertex is \textit{eligible} at some time point if all its predecessors in the current release have been finished and thus it can immediately execute if there are available processors.
We assume each DAG has a unique head vertex (with no predecessor) and a unique tail vertex (with no successor). This assumption does not limit the expressiveness of our model since one can always add a dummy head/tail vertex to a DAG with multiple entry/exit points.

%

\begin{figure}[!t]
	\centering
	\includegraphics[width=2.4in]{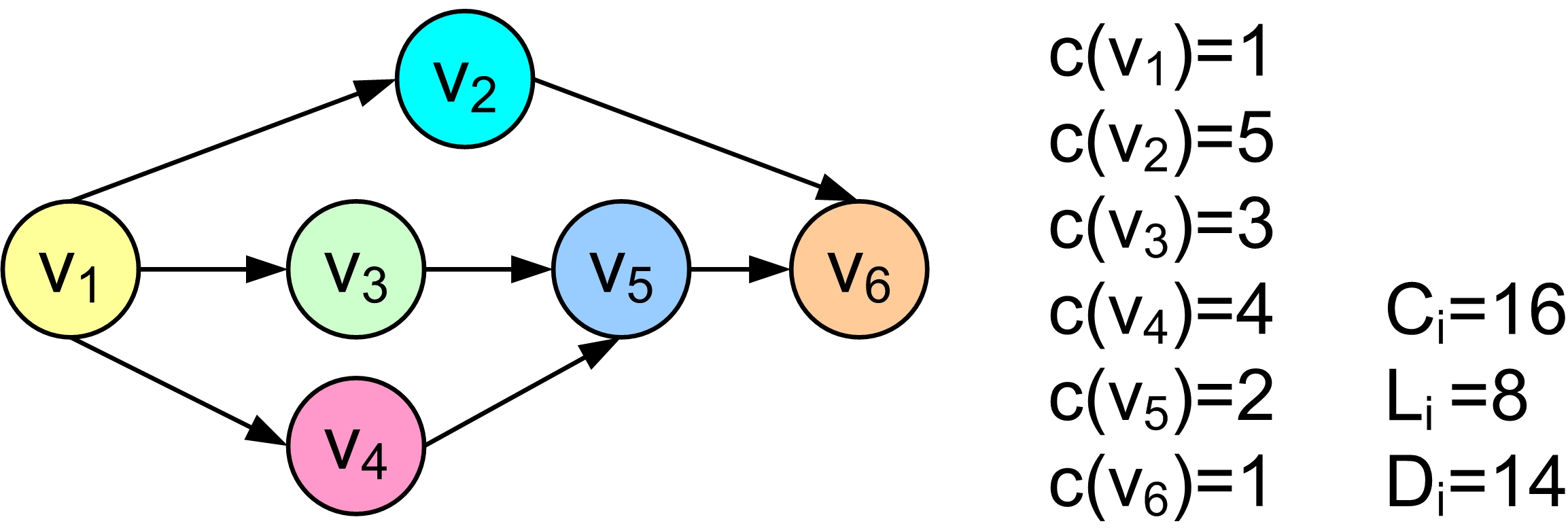}
	\caption{\xuadd{A DAG task example.}}
	\label{fig1}
\end{figure}

We use $C_i$ to denote the total worst-case execution time of all vertices of DAG task $\tau_i$ and $L_i$ to denote the sum of $c(v)$ of each vertex $v$ along the longest chain (also called the critical path) of $\tau_i$.
The \textit{utilization} of a DAG task $\tau_i$ 
is  $\frac{C_i}{T_i}$, and its
\textit{density} is $\frac{C_i}{D_i}$.
A DAG task 
is called a \textit{heavy} task if it density
is larger than $1$, and a \textit{light} task otherwise.

Figure \ref{fig1} shows a DAG task with $6$ vertices. We can compute $C_i = 16$ 
and $L_i = 8$ (the longest path is $\{v_1, v_4, v_5 , v_6\}$). This is a heavy task since the density is $\frac{16}{14} > 1$.




%

\subsection{Federated Scheduling}

In federated scheduling \cite{li2014analysis}, each heavy task exclusively executes
on a subset of dedicated processors. Light tasks are treated as traditional sequential real-time tasks and share the remaining processors. As a heavy task exclusively owns several dedicated processors and its workload must be finished before the next release time (due to constrained deadlines), the response time of a heavy task can be bounded using the classical result for non-recurrent DAG tasks by Graham \cite{graham1969bounds}:
\begin{equation}\label{e:fedbound}
R_i \leq L_i+\frac{C_i-L_i}{m_i}
\end{equation}
where $m_i$ is the number of dedicated processors granted to this heavy task. Therefore,  by setting $R_i = D_i$, we can calculate the minimal amount of processing capacity required by this task to meet its deadline $\frac{C_i - L_i}{D_i - L_i}$, and the number of processors assigned to a heavy task $\tau_i$ is the minimal integer no smaller than 
$\left \lceil\frac{C_i-L_i}{D_i-L_i} \right\rceil$.
The light tasks are treated as sequential sporadic tasks, and are scheduled on the remaining processors by traditional multiprocessor scheduling algorithms, such as global EDF \cite{baruah07rtss} and partitioned EDF \cite{baruah05rtss}.

\section{A Single DAG on Uniform Multiprocessors}\label{s:single}

In this section, we focus on the problem of bounding the response time of a single DAG task exclusively executing on a \textit{uniform} multiprocessor platform, where processors in general have different speeds.
The reason why we study the case of uniform multiprocessors is as follows.
In the semi-federated scheduling, a heavy task may share processors with others. From this heavy task's point of view, it only owns a portion of the processing capacity of the shared processors. Therefore, to analyze semi-federated scheduling, we first need to solve the fundamental problem of how to bound the response time in the presence of portioned processing capacity.
The results of this section 
serve as the theoretical foundation for semi-federated scheduling on \textit{identical} multiprocessors in later sections (while they
also can  be directly used for federated scheduling on uniform multiprocessors as a byproduct of this paper).

\subsection{Uniform Multiprocessor Platform}
We assume a uniform multiprocessor platform of 
$m$ processors, characterized by their normalized speeds
$\{\speed_1, \speed_2, \cdots, \speed_{m}\}$.
Without loss of generality, we assume the processors are sorted in non-increasing speed order ($\speed_x \geq \speed_{x+1}$) and the fastest processor has a unit speed i.e., $\speed_1 = 1$.
In a time interval of length $t$, the amount of workload executed on a processor with speed $\speed_x$ is $t \speed_x$.
Therefore, if the WCET of some workload on a unit speed processor is $c$, then its WCET becomes $c/\speed_{x}$ on a processor  of speed $\speed_{x}$.




\subsection{Work-Conserving Scheduling on Uniform Multiprocessors}\label{ss:work-conserving}
On identical multiprocessors, a work-conserving scheduling algorithm never leaves any processor idle while there exists some eligible vertex. 
The response time bound for a DAG task in (\ref{e:fedbound}) is applicable to \textit{any} work-conserving scheduling algorithm, regardless what particular strategy is used to assign the eligible vertices to available processors.

However, on uniform multiprocessors, the strategy to assign eligible vertices to processors may greatly affect the timing behavior of the task. Therefore, we extend the concept of work-conserving scheduling to uniform multiprocessors by enforcing execution on faster processors as much as possible \cite{funk2001line}. 
More precisely, a scheduling
algorithm is work-conserving on $m$ uniform processors if it satisfies both of the following conditions:
\begin{enumerate}
\item No processor is idled when there are eligible vertices awaiting execution.

\item If at some time point there are fewer than $m$ eligible vertices awaiting execution,
then the eligible vertices are executed upon the fastest processors.
\end{enumerate}

Figure \ref{fig:uniformworkconserving} shows a possible scheduling sequence of the DAG task on three processors with speeds $\{1, 0.5, 0.25\}$. Vertex $v_2$ migrates  to the fastest processor at time $5$ and $v_5$ migrates to the fastest processor at time $9$.
These two extra migrations
are the price paid for satisfying the second condition of work-conserving scheduling in above.

If we disallow the migration from slower processors to faster processors, there may be significant resource waste. In the worst case, a DAG task will execute its longest path on the lowest processor, which results in very large response time. 
In Appendix-A we discuss the response time bound and resource waste when the inter-processor migration is forbidden.

\begin{figure}[!t]
	\centering
\includegraphics[width=3.2in]{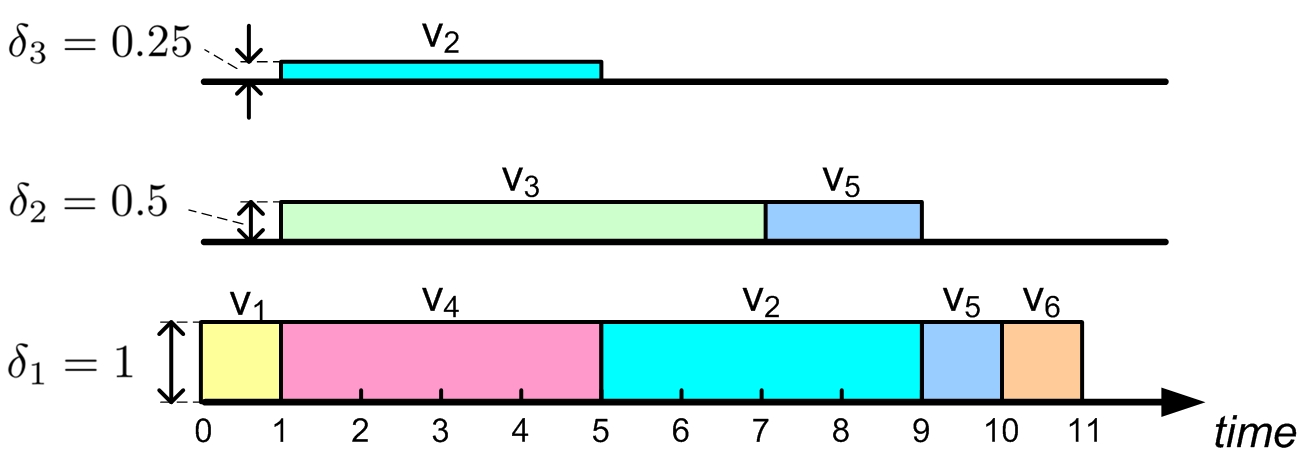}
	\caption{A work-conserving scheduling sequence on uniform multiprocessors.} 
	\label{fig:uniformworkconserving}
\end{figure}

\subsection{Response Time Bound}
In the following we derive response time bounds for a single DAG task executing on a uniform multiprocessor platform under work-conserving scheduling.
Although the task is recurrent, we only need to analyze its behavior in one release since the task has a constraint deadline.
We first introduce the concept \textit{uniformity} \cite{funk2001line}: 
\begin{define}[Uniformity]
The \textit{uniformity} of 
$m$ processors
with speeds $\{\speed_1, \cdots, \speed_{m} \}$ ($\speed_x \geq \speed_{x+1}$) is defined as 
\begin{equation}
\platfeature \stackrel{\Delta}{=}
\max_{x=1}^{m}\left\lbrace \frac{\sumspeed_m-\sumspeed_x}{\speed_x}\right\rbrace 
  \label{eq:uniformality}
\end{equation}
where $\sumspeed_x$ is the sum of the speeds of the $x$ fastest processors:
\begin{equation}
  \sumspeed_x\stackrel{\Delta}{=}\sum_{j=1}^{x}\speed_j
  \label{eq:defineS}
  \end{equation}
  \label{define:uniformality}
\end{define}
\vspace{-0.1in}
Now we derive the response time upper bound:
\begin{theorem}
The response time of a DAG task $\tau_i$  executing on $m$ processors
with speeds $\{\speed_1, \cdots, \speed_{m}\}$ 
is bounded by:	
	\begin{equation}
	\response\leq\frac{C_i +\platfeature  L_i}{\sumspeed_m}
	\label{eq:singlebound}
	\end{equation}	 
where $\lambda$ and $S_m$ are defined in Definition   \ref{eq:uniformality}.
	\label{th:singlebound}
\end{theorem}
\begin{figure}[!t]
	\centering
	\includegraphics[width=3.2in]{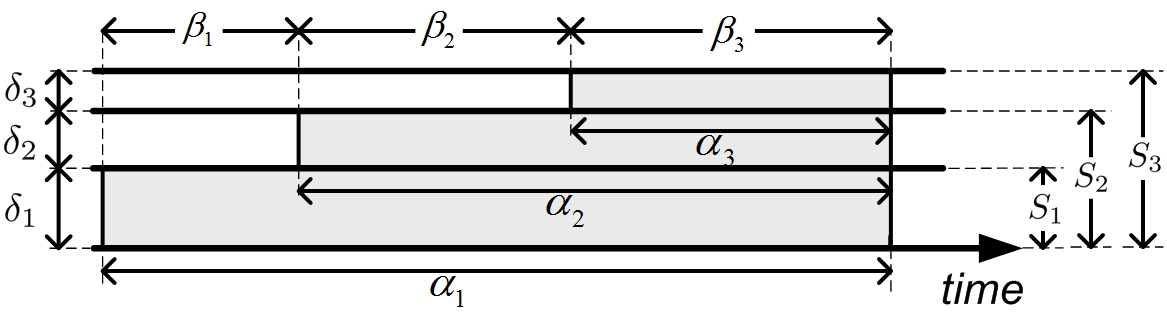}
	\caption{Illustration of $\alpha_x$ and $\beta_x$.}
	\label{fig:example3}
\end{figure}
\begin{proof}
For simplify of presentation, we assume that each vertex $v$ executes exactly for its WCET $c(v)$\footnote{It is easy to show that the response time bound in (\ref{eq:singlebound}) still holds if some vertices execute for shorter than its WCET.}.
Without loss of generality, we assume the task
under analysis releases an instance at time $0$, and
thus finishes the current
instance at time 
$\response$. During the time window $[0,\response]$, let $\alpha_x$ denote the total length of intervals during which the $x^{th}$ processor (with speed $\speed_x$) is busy. By the work-conserving scheduling rules in Section \ref{ss:work-conserving}, we know if the $x^{th}$ processor is busy in a time interval then all faster processors 
(with index smaller than $x$) must also be busy. Therefore, we know $R = \alpha_1$. We define 

\begin{equation*}
\beta_x = \left \{
\begin{array}{cl}
\alpha_x-\alpha_{x+1},                        & ~~ 1\leq{x}< m  \\
\alpha_x, & ~~x= m
\end{array}\right.
\end{equation*}
Figure \ref{fig:example3} illustrates the definition of $\alpha_x$ and $\beta_x$.
So we can rewrite $R=\alpha_1$ as
\begin{equation}
\response=\sum_{x=1}^{m}\beta_x 
\label{eq:R=}
\end{equation}	
The total workload executed on all the processors in $[0, \response]$ is
$(\beta_1 S_1 + \cdots + \beta_{m} S_{m})$, which equals the total worst-case execution time of the task:
\begin{equation}
C_i = \sum_{x=1}^{m}\beta_x  \sumspeed_x
\label{eq:C=}
\end{equation}


Let $v_z$ be the latest finished vertex in the DAG and $v_{z-1}$ be the latest finished vertex among all predecessors of $v_z$, repeat this way until no predecessor can be found, we can simply construct a path $\pi=\{v_1,v_2,...,v_{z-1},v_z\}$. The fundamental observation is that all processors must be busy between the finishing time of $v_k$ and the starting time of $v_{k+1}$ where $1\leq{k}\leq{z-1}$.  We use
$\consume(\pi, \speed_x)$ to denote the total amount of workload executed for vertices along path $\pi$ in all the time intervals during which both of the following conditions are satisfied:
\begin{itemize}
\item at least one processor is idle
\item the slowest busy processor has speed $\speed_x$.
\end{itemize}

The total length of such time intervals is $\beta_x$.
Since at least one processor is idle, $\pi$ must contain a vertex being executed in this time interval (since at any time point before $R$, there is at least one eligible vertex along $\pi$). 
So we have
\begin{align}
  \consume(\pi, \speed_x) &\geq \beta_x \speed_x \notag\\
\Rightarrow  ~~ \sum_{x=1}^{m-1}
\consume(\pi, \speed_x) & \geq \sum_{x=1}^{m-1} \beta_x  \speed_x
\label{e:guan-2}
\end{align}
Let $l_\pi$ denote the total workloadalong path $\pi$, 
so we know
\[
\sum_{x=1}^{m-1}
\consume(\pi, \speed_x) \leq l_\pi
\]
Since $L_i$ is the total workload of the longest path in the DAG, we know $l_\pi \leq L_i$. In summary, we have
\begin{equation}
\sum_{x=1}^{m-1}
\consume(\pi, \speed_x)  \leq L_i
\label{e:guan-3}
\end{equation}

Combining (\ref{e:guan-2}) and (\ref{e:guan-3}) gives
\begin{equation}
 L_i \geq \sum_{x=1}^{m-1} \beta_x \speed_x 
~\Rightarrow~
\platfeature L_i   \geq \sum_{x=1}^{m-1} \beta_x  \speed_x  \platfeature
\label{e:guan-1}
\end{equation}

By the Definition of $\platfeature$ in (\ref{eq:uniformality}) we know
\begin{equation*}
\forall x: \frac{\sumspeed_{m}-\sumspeed_x}{\speed_x}\leq\platfeature
\end{equation*}	
Therefore, (\ref{e:guan-1}) can be rewritten as
\begin{equation*}
 \platfeature L_i  \geq \sum_{x=1}^{m-1}\beta_x
(\sumspeed_{m}-\sumspeed_x) 
\end{equation*}
and by applying (\ref{eq:C=}) we get 
\begin{align*}
& C_i + \platfeature L_i   \geq \sum_{x=1}^{m-1}\beta_x
(\sumspeed_{m}-\sumspeed_x) + \sum_{x=1}^{m}\beta_x  \sumspeed_x \\
\Leftrightarrow ~~&
C_i + \platfeature L_i   \geq \beta_{m} \sumspeed_{m} + \sum_{x=1}^{m-1}\beta_x\sumspeed_{m}  \\
\Leftrightarrow ~~&
C_i +  \platfeature L_i  \geq 
\sumspeed_{m} \sum_{x=1}^{m}\beta_x
\end{align*}
and by applying (\ref{eq:R=}), the theorem is proved.
\end{proof}
When $\speed_1 =  \cdots = \speed_m = 1$, we have $\platfeature = m-1$ and $S_{m} = m$, so the bound in Theorem 
\ref{th:singlebound}
perfectly degrades to (\ref{e:fedbound})
for the case of identical processors.

\section{Runtime Dispatcher of Each DAG}\label{s:framework}

The conceptual uniform multiprocessor platform in last section imitates the resource obtained by a task when sharing processors with other tasks.
A naive way to realize the conceptual uniform multiprocessors on our identical unit-speed multiprocessor platform is to use fairness-based scheduling, in which
task switching is sufficiently frequent so that each task receives a fixed portion of processing capacity.
However, this approach incurs extremely high context switch overheads which may not be acceptable in practice.

In the following, we introduce our method to realize the proportional sharing of processing capacity without frequent context switches. The key idea is to use a runtime \textit{dispatcher} for each DAG task to encapsulate the execution 
on a conceptual processor with speed $\delta_p$ into a \textit{container task} $\varphi_p$ with a \textit{load bound} $\delta_p$.
The dispatcher guarantees that the  workload encapsulated into a container task does not exceed its load bound.
These container tasks are scheduled using priority-based scheduling algorithms and their load bounds can be used to decide the schedulability.

As will be introduced in the next section, in our semi-federated scheduling algorithms,
most of the container tasks used by a DAG task have a load bound $1$, 
which correspond to the dedicated processors, and only 
a few of them have fractional load bounds ($<1$). 
However, for simplicity of presentation, in this section we treat all container tasks identically, regardless whether 
the load bound is $1$ or not.

Suppose we execute a DAG task through $m$ container tasks $\{\varphi_1, \varphi_2, \cdots, \varphi_m\}$.
Each of the container task is affiliated with the following information $\varphi_p = (\delta_p, d_p, exe_p)$:
\begin{itemize}
	\item $\delta_p$: the load bound of $\varphi_p$, which is a fixed value.
	\item $d_p$: the absolute deadline of $\varphi_p$, which varies at runtime.
	\item $exe_p$: the vertex currently executed by $\varphi_p$, which also varies at runtime
\end{itemize}


%
%
%


At each time instant, a container task is either \textit{occupied} by some vertex or \textit{empty}. If a container task is occupied by vertex $v$, i.e., $exe_p = v$, then 
this container task is responsible to execute the workload of $v$
and the maximal workload executed by this container task executes
before the absolute deadline $d_p$ is $c(v)$. A vertex $v$ may be divided into several parts, and the their total WCET equals $c(v)$, as will be discussed later when we introduce Algorithm \ref{a:dispatcher}.
\textbf{Note that an occupied container task becomes empty when time reaches its absolute deadline}. 

\begin{algorithm}
	\caption{The dispatching algorithm (invoked at time $t$).}
	\begin{algorithmic}[1]	
		\STATE $v =$ an arbitrary eligible vertex in $S$ ($S$ stores the set of vertices that have not been executed yet); 
		\STATE Remove $v$ from $S$; 
		\STATE $\varphi_p =$ the empty container task with the largest load bound; 	
		\STATE $d' = $ the earliest deadline of all occupied container tasks with load bound strictly larger than $\delta_p$;
		\IF{\!$(\textrm{all container tasks are empty}) 
			\vee (d'  >  t + \frac{c(v)}{\delta_p})$} \label{line:if}
		\STATE $d_p =  t + c(v) / \delta_p$ \label{line:5}
		\STATE $exe_p = v$  \label{line:6}
		\ELSE
		\STATE $d_p =  d'$  \label{line:8}
		\STATE Split $v$ into $v'$ and $v''$ so that \label{line:split}
		\[c(v') = (d_p - t) \times \delta_p ~\textrm{and}~ c(v'') = c(v) - c(v')\] 
		\STATE $exe_p = v'$
		\STATE Put $v''$ back to the head of $S$; 
		\STATE Add a precedence constraint from $v'$ to $v''$; \label{line:12}
		\ENDIF
	\end{algorithmic}
	\label{a:dispatcher}
\end{algorithm}

The pseudo-code of the dispatcher is shown in Algorithm \ref{a:dispatcher}.
At runtime, the dispatcher is invoked when there exist both empty container tasks and eligible vertices.
The target of the dispatcher is to assign (a part of) an eligible vertex to the \textit{fastest} empty container task.

The absolute deadline $d_p$ of a container task $\varphi_p$ mimics the finishing time of a vertex if it is executed on a processor with the speed $\delta_p$.
When the container task starts to be occupied by a vertex $v$ at time $t$,
$d_p$ is set  to be $d_p = t + c(v)/\delta_p$. 
Therefore, we have the following property of Algorithm \ref{a:dispatcher}. First, 
the dispatcher guarantees the execution rate of a container task is consistent with the corresponding uniform processors:
\begin{property}
If $\varphi_p$ starts to be occupied by $v$ from  $t_1$ and becomes empty at $t_2$, the maximal workload
executed by $\varphi_p$ in $[t_1, t_2)$ is $(t_2 - t_1) \delta_p$.
\end{property}

Another key point of Algorithm \ref{a:dispatcher} is 
always 
keeping the container task with larger load bounds being occupied, 
which mimics the second work-conserving scheduling rule on uniform multiprocessors (workload is always executed on faster processors). This is done by checking the condition in line \ref{line:if}:
\begin{equation}\label{e:splitcond}
d'  >  t + c(v)/\delta_p
\end{equation}
where $d'$ is the earliest absolute deadline among all the container tasks currently being occupied and $\delta_p$ is load bound of the fastest empty container task which will be used now. If this condition does not hold, 
putting the entire $v$ into $\varphi_p$ may lead to the situation that 
a container task with a larger load bound becomes empty while $\varphi_p$ 
is still occupied. This corresponds to the situation on uniform processors that a faster processor is idle while a slower processor is busy, which violates the second work-conserving scheduling rule. To solve this problem, 
in Algorithm \ref{a:dispatcher}, when condition (\ref{e:splitcond}) does not hold,
$v$ is split into two parts $v'$ and $v''$, so that $\varphi_p$ only executes
the first part $v'$, whose deadline exactly equals to 
the earliest absolute deadline of all faster container tasks (line \ref{line:split}).
The remaining part $v''$ is put back to $S$ and will be assigned in the future, and a precedence from $v'$ to $v''$ is established to guarantee that $v''$ become eligible only if $v'$ has finished. 
In summary, Algorithm \ref{a:dispatcher} guarantees the following property: 
\begin{property}
The eligible vertices are always executed upon the container tasks with the largest load bounds.
\end{property}

\begin{figure}[!t]
	\centering
\includegraphics[width=3.2in]{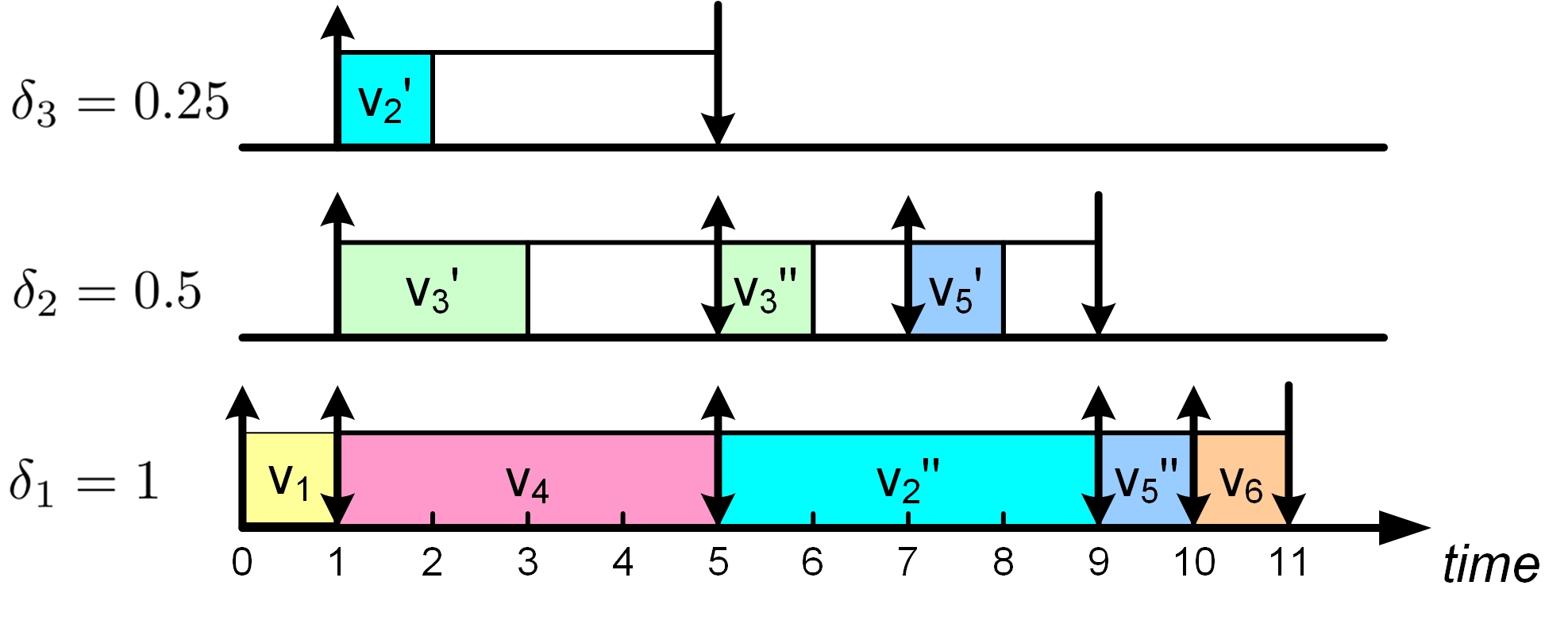}
	\caption{A scheduling sequence on container tasks.} 
	\label{fig:deadlineassignment}
\end{figure}

Figure \ref{fig:deadlineassignment} shows a possible scheduling sequence of the example DAG task in Figure \ref{fig1} 
executed on three container tasks with load bounds $\delta_1 = 1$, $\delta_2 = 0.5$ and $\delta_3 = 0.25$.
An upwards arrow represents an empty container task becoming occupied and a downwards arrow represents an occupied task becoming empty.
Algorithm \ref{a:dispatcher} is invoked whenever there exist both eligible 
vertices and empty container tasks. 
This scheduling sequence corresponds to the scheduling sequence 
of the same task on uniform processors with speeds $\delta_1 = 1$, $\delta_2 = 0.5$ and $\delta_3 = 0.25$ in Figure 
\ref{fig:uniformworkconserving}. We can see that the amount of workload executed between any two time points at which Algorithm \ref{a:dispatcher} is invoked, is the same in both scheduling sequences. An step-by-step explanation of 
this example is given in Appendix-B.

In general, 
if each container task always finishes the workload of its assigned vertex before the corresponding deadline, 
the scheduling sequence resulted by Algorithm \ref{a:dispatcher} 
on container tasks with load bounds 
$\{\delta_1, \cdots, \delta_{m}\}$
corresponds to 
a work-conserving scheduling sequence of the same DAG task on uniform multiprocessors
with speeds $\{\delta_1, \cdots, \delta_{m}\}$. Therefore the response time bound in Theorem \ref{th:singlebound}
can be applied to bound the response time of the DAG task executed on container tasks using Algorithm \ref{a:dispatcher}. By the above discussions, we can conclude
the following theorem.

\begin{theorem}\label{t:framework}
	Suppose a DAG task $\tau_i$ executes on
	$m$ container tasks with load bounds $\{\speed_1, \cdots, \speed_{m}\}$ and each container task always
	finishes its assigned workload before the corresponding 
	absolute deadline, then the response time $R$ of $\tau_i$
	is upper bounded by:	
	\begin{equation}
	\response\leq\frac{C_i +\platfeature  L_i}{\sumspeed_{m}}
	\label{eq:framework}
	\end{equation}	 
\end{theorem}

\section{Semi-Federated Scheduling Algorithms} \label{s:semi_federated}

In this section, we propose two semi-federated scheduling algorithms based on container task and runtime dispatcher introduced in last section. In the first algorithm, a DAG task requiring $x + \epsilon$ processing capacity is granted $x$ dedicated processors and \textit{one} container task with load bound $\epsilon$, and all the container tasks and the light tasks are scheduled by \textit{partitioned} EDF on the remaining processors. The second algorithm enhances the first one by allowing to divide the fractional part $\epsilon$ into \textit{two} container tasks, which further improves resource utilization.

\subsection{The First Algorithm: \firstSFS}

By Theorem \ref{t:framework} we know 
a DAG task is schedulable if
the load bounds $\{\speed_1, \cdots, \speed_{m}\}$  of the container tasks satisfy
\begin{equation}\label{e:check-1}
\frac{C_i + \platfeature L_i}{S_m} \leq D_i
\end{equation}
where $\platfeature$ is the uniformity and $S_m$ is the sum of $\{\speed_1, \cdots, \speed_{m}\}$, defined in Definition \ref{define:uniformality}. 
There are difference choices of the container tasks to make a DAG task schedulable.
In general, we want to make the DAG task to be schedulable with as little processing capacity as possible.
The load bound of a container task actually represents its required processing capacity, and thus $S_{m}$ represents the total processing capacity required by all the container tasks 
of a DAG task. 
In the following, we will introduce how to choose the feasible container task set with the minimal $S_{m}$.

We first show that the total load bound of any container task set that can pass the condition (\ref{e:check-1}) has a lower bound:

\begin{define}
	The \emph{minimal capacity requirement} $\gamma_i$ of a DAG task $\tau_i$
	is defined as:
	\begin{equation}\label{e:gamma}
	\gamma_i  = \frac{C_i - L_i}{D_i - L_i}
	\end{equation} 
\end{define}

\begin{lemma}\label{l:lowerbound}
	A DAG task $\tau_i$ is scheduled on $m$ container tasks with load bounds $\{\delta_1, \delta_2, \cdots, \delta_m \}$. 
	If condition (\ref{e:check-1}) is satisfied,
	then it must hold
	\[
	S_{m} \geq \gamma_i 
	\]
\end{lemma}

\begin{proof}
	Without loss of generality, we assume the container tasks
	are sorted in non-increasing order of their load bounds, i.e.,
	$\delta_p \geq \delta_{p+1}$. 	
	By the definition of $\platfeature$ we have
\[
\platfeature \geq \frac{S_{m} - \delta_1}{\delta_1} 
\]	
and since the load bounds are at most $1$, i.e., $\delta_1 \leq 1$, we know
\[
\platfeature \geq S_{m} - 1
\]	
Applying this to (\ref{e:check-1}) yields 
\begin{align*}
 \frac{C_i + (S_m - 1)L_i}{S_m} \leq D_i \Rightarrow  S_m \geq \frac{C_i - L_i}{D_i - L_i}
\end{align*}
so the lemma is proved.
\end{proof}

Next we show that the minimal capacity requirement is achieved by 
using only one container task with a fractional load bound ($<1$) and $x$ container tasks with load bound $1$: 

\begin{lemma}\label{l:mplusone}
A DAG task $\tau_i$ is schedulable on $x$ container tasks with load bound of $1$ and one container task with load bound $\delta$, where
$x =  \lfloor \gamma_i  \rfloor$ and $\delta = \gamma_i  -  \lfloor \gamma_i  \rfloor$.
\end{lemma}

\begin{proof}
By the definition of $\platfeature$, we get
\[
\platfeature {=}
 \max\left(
 \frac{\gamma_i - 1}{1}, \cdots, \frac{\gamma_i - \lfloor \gamma_i \rfloor}{1} \right) = \gamma_i -1 \\
\]
and we know $S_m = x + \delta = \gamma_i$. 
So by (\ref{eq:framework}) the response time of $\tau_i$ is bounded by
\[
R \leq  \frac{C_i + (\gamma_i -1) L_i}{\gamma_i}
\]
In order to prove $\tau_i$ is schedulable, it is sufficient to prove
\[
 \frac{C_i + (\gamma_i -1) L_i}{\gamma_i} \leq D_i
\]
which must be true by the definition of $\gamma_i$.
\end{proof}

In summary, by Lemma \ref{l:lowerbound} and \ref{l:mplusone}
we know using $x$ container tasks with load bound $1$ and \textit{one} container task with a fractional load bound requires the minimal
processing capacity, which motivates our first scheduling algorithm $\firstSFS$.

 \begin{algorithm}
 	\caption{The first semi-federated algorithm: \firstSFS.}
 	\begin{algorithmic}[1]
 		\FOR{each heavy task $\tau_i$}
 		\STATE $\gamma_i = \frac{C_i - L_i}{D_i - L_i}$ \label{line:2}
 		\IF{less than $\lfloor \gamma_i \rfloor$ avaiable processors}
 		\STATE \textbf{return} \textit{failure}
 		\ENDIF
 		\STATE assign  $\lfloor \gamma_i \rfloor$ dedicated processors to $\tau_i$
 		\STATE create a container task with load bound $ \gamma_i - \lfloor \gamma_i \rfloor$ for $\tau_i$\label{line:7}
 		\ENDFOR
 		\STATE $\Omega =$ the set of remaining processors
 		\STATE $S =$ the set of container tasks $\cup$ the set of light tasks
 		\STATE \textbf{if} \Sched($S$, $\Omega$) \textbf{then} \textbf{return} \textit{success} \textbf{else return} \textit{failure}\label{line:11}
 	\end{algorithmic}
 	\label{al:first}
 \end{algorithm}

The pseudo-code of $\firstSFS$ is shown in Algorithm  	\ref{al:first}. 
The rules of \firstSFS\ can be summarized as follows:
\begin{itemize}
	\item Similar to the federated scheduling, \firstSFS\ also classifies the DAG tasks into heavy tasks (density $>1$)
	and light tasks (density $\leq 1$).
	
	\item For each heavy task $\tau_i$, we grant $\lfloor \gamma_i \rfloor$ dedicated processors and one container task
	with load bound $\gamma_i - \lfloor\gamma_i \rfloor$ to it
	where $\gamma_i = \frac{C_i - L_i}{D_i - L_i}$
	(line \ref{line:2} to \ref{line:7}). The algorithm declares a failure if some heavy tasks cannot get enough dedicated processors.

	\item After granting dedicated processors and container tasks to all heavy tasks, 
	the remaining processors will be used to schedule the light tasks and container tasks. 
	The function \Sched($S$, $\Omega$) (in line \ref{line:11}) returns the
 the schedulability testing result
	of the task set consisting of light tasks and container tasks on processors in $\Omega$.
\end{itemize}
Various multiprocessor scheduling algorithms can be used to schedule the light tasks and container tasks, such as global EDF and partitioned EDF. In this work, we choose to use partitioned EDF, and in particular with the
Worst-Fit packing strategy \cite{johnson1974worst}, to schedule them.

More specificly, 
at design time, the light tasks and container tasks 
are partitioned to the processors in $\Omega$.
Tasks are partitioned  
in the non-increasing order of their load (the load of a light task $\tau_i$ equals its density $C_i/D_i$, and the load of a container task $\varphi_p$ equals its load bound $\delta_p$). 
At each step the processor with the minimal 
total load of currently assigned tasks is selected, as long the total 
load of the processor after accommodating this task still does not 
exceed $1$. \Sched($S$, $\Omega$) returns \textit{true} if all tasks are partitioned to some processors, and returns \textit{false} otherwise.

At runtime, the jobs of tasks partitioned to each processor are
scheduled by EDF.
Each light task behaves as a standard sporadic task. Each container task behaves as a
GMF (general multi-frame) task  \cite{baruah99gmf}:
when a container task $\varphi_p$ starts to be occupied by a vertex $v$, $\varphi_p$ releases a job with WCET $c(v)$ and an absolute deadline $d_p$ calculated by Algorithm \ref{a:dispatcher}. Although a container task $\varphi_p$ releases different types of jobs, its load is bounded by $\delta_p$
as the \textit{density} of each of its jobs is $\delta_p$.

Appendix-C presents an example to illustrate
\firstSFS.

Recall that in the runtime dispatching, a vertex may be split into two parts, in order to guarantee a \enquote{faster} container task is never empty when a 
\enquote{slower} one is occupied. The following theorem 
bounds the number of extra vertices created due to the splitting in \firstSFS.

\begin{theorem}\label{t:firstoverhead}
Under \firstSFS,
the number of extra vertices created in each DAG task is bounded by  the number of vertices in the original DAG. 
\end{theorem}
\begin{proof}
	Let $N$ be the number of vertices in the original DAG. According to Algorithm	\ref{a:dispatcher}, a vertex will not be split if it is
dispatched to a dedicated processor (i.e., a container task with load bound $1$). The number of vertices executed on these dedicated processors is at most $N$.
A vertex my be split when being dispatched to 
the container task with a fractional load bound, and upon each splitting, the deadline of the first part must align with 
some vertices on the dedicated processors, so the number of splitting is bounded by $N$.
\end{proof}


%


\subsection{The Second Algorithm: \secondSFS}

In partitioned EDF, \enquote{larger} tasks in general
lead to worse resource waste.
The system schedulability can be improved if
tasks can be divided into small parts.
In \firstSFS, each heavy task is granted several dedicated processors and \textit{one} container task with fractional load bound. The following examples shows we can actually divide this container task into \textit{two} smaller ones without increasing the total processing capacity requirement.

Consider the DAG task in Figure \ref{fig1},
the minimal capacity requirement of which is
\[
\gamma_i = \frac{C_i - L_i}{D_i - L_i}  = \frac{16 - 8}{14 - 8} = \frac{4}{3}
\]
Accordingly, \firstSFS\ assigns one dedicated processor and one container task with load bound $\frac{1}{3}$ to this task.

Now we replace the container task
with load bound $\frac{1}{3}$ 
 by two container tasks with load bounds $\frac{1}{4}$ and $\frac{1}{12}$.
After that, the total capacity requirement is unchanged since $\frac{1}{3} = \frac{1}{4} + \frac{1}{12}$, and the DAG task is still schedulable since the uniformity of both $\{1, \frac{1}{3}\}$ and $\{1, \frac{1}{4}, \frac{1}{12}\}$
is $\frac{1}{3}$.

However, in general dividing a container task into two may increase the uniformity. 
For example, if we divide the container task in the above example into two container tasks both with load bound $\frac{1}{6}$, the uniformity 
is increased to $1$ and the DAG task is not schedulable. The following lemma gives the condition for dividing one container task into two without increasing the uniformity:
\begin{lemma}\label{l:dividetwo}
A heavy task $\tau_i$ with minimal capacity requirement $\gamma_i$ is scheduled on
$\lfloor \gamma_i \rfloor$ dedicated processors and two container tasks with load bounds $\delta'$ and $\delta''$ s.t.
\[
\delta' + \delta'' = \gamma_i - \lfloor \gamma_i \rfloor
\]
$\tau_i$ is schedulable if
\begin{equation}\label{e:dividetwo-0}
\delta' \geq \max\left(\frac{\gamma_i - \lfloor \gamma_i \rfloor}{2}, \frac{\gamma_i - \lfloor \gamma_i \rfloor}{\gamma_i}\right)
\end{equation}
\end{lemma}
\begin{proof}
By Theorem \ref{t:framework}
we know the response time of $\tau_i$ is bounded by
\begin{equation}\label{e:dividetwo-1}
R \leq \frac{C_i+\platfeature{L_i}}{S_m}
\end{equation}
Since $\delta' + \delta'' =  \gamma_i - \lfloor \gamma_i \rfloor $
and $\delta' \geq ( \gamma_i - \lfloor \gamma_i \rfloor)/2 $, 
we know $\delta' \geq \delta''$.
So we can calculate $\platfeature$ 
of $\lfloor \gamma_i \rfloor$ dedicated processors and
two container tasks with load bounds $\delta'$ and $\delta''$ by:
\begin{align}
\platfeature & {=}
\max_{x=1}^{m}\left\lbrace \frac{\sumspeed_m-\sumspeed_x}{\speed_x}\right\rbrace \notag\\
 & {=}
\max\left(
\frac{\gamma_i - 1}{1}, \frac{\gamma_i - 2}{1}, \cdots, \frac{\gamma_i - \lfloor \gamma_i \rfloor}{1}, \frac{\delta''}{\delta'}, \frac{0}{\delta''} \right) \notag \\
 & {=}
 \max\left(
 \frac{\gamma_i - 1}{1}, \frac{\delta''}{\delta'} \right)  \label{e:divide-3}
\end{align}
By $\delta' + \delta'' =  \gamma_i - \lfloor \gamma_i \rfloor $ and
$\delta' \geq \frac{  \gamma_i - \lfloor \gamma_i \rfloor}{\gamma_i}$
we get $\frac{\delta''}{\delta'} \leq \gamma_i - 1$. Applying this
to (\ref{e:divide-3}) gives $\lambda = \gamma_i - 1$.
Moreover, we know $S_m = \lfloor \gamma_i \rfloor +  \delta' + \delta'' = \gamma_i$. Therefore, we
have
\[
R \leq \frac{C_i+\platfeature{L_i}}{S_m}=\frac{C_i+(\gamma_i-1)L_i}{\gamma_i}
\]
and by the definition of $\gamma_i$ in (\ref{e:gamma}) we know 
\[
\frac{C_i+(\gamma_i-1)L_i}{\gamma_i} = D_i
\]
so we can conclude $R_i \leq D_i$, and thus $\tau_i$ is schedulable.
\end{proof}

%

Based on the above discussions, we propose the second federated scheduling algorithm \secondSFS. The overall procedure of \secondSFS\ is similar to \firstSFS. The only difference is that
\secondSFS\ uses \Sched$^*$($S$, $\Omega$) to replace 
\Sched($S$, $\Omega$) in line \ref{line:11} of Algorithm  	\ref{al:first}. The pseudo-code of \Sched$^*$($S$, $\Omega$) 
is given in Algorithm  \ref{al:second}.
The inputs of \Sched$^*$ are $S$, the set of 
sequential tasks (including the generated container tasks
and the light tasks), and $\Omega$, the remaining processors to be shared by these sequential tasks.

There are infinitely many choices to divide a container task into two under the condition of Lemma \ref{l:dividetwo}.
Among these choices, on one simply dominates  others, since the quality of a choice
depends on how the tasks are partitioned to processors.
In \Sched$^*$($S$, $\Omega$), the container tasks are divided in an on-demand manner.
Each container task $\varphi_k$ of task $\tau_i$, apart from
its original load bound $\delta_k$, is
affiliated with a ${\delta_k^*}$, 
representing the minimal load bound of
the larger part if $\varphi_k$ is divided into two parts.
${\delta_k^*}$ is calculated according to Lemma \ref{l:dividetwo}:
\begin{align}\label{e:lowerupper}
{\delta_k^*} & = \max\left(\frac{\gamma_i - \lfloor \gamma_i \rfloor}{2},  \frac{\gamma_i - \lfloor \gamma_i \rfloor}{\gamma_i} \right)
\end{align}
For consistency, each light task $\tau_j$ is also affiliated with a ${\delta_j^*}$ which equals to its density $\delta_j = C_i/D_i$.

\Sched$^*$($S$, $\Omega$) works in three steps:

\begin{enumerate}
	\item It first partitions all the input container tasks and light tasks using the Worst-First packing strategy 
	using their ${\delta_k^*}$ as the metrics. 
	We use $\varphi(P_x)$ to denote the set of tasks have been assigned to processor $P_x$.
	If the sum of $\delta_k$ of all tasks in $\varphi(P_x)$ has  exceeded $1$, we stop assigning tasks to $P_x$ and move it to the set $\Psi$. 
	
	\item The total $\delta_k$ of tasks on each processor  $P_x$ in $\Psi$  is larger than $1$, therefore some of tasks on  $P_x$ must be divided into two, and one of them should be assigned to other processors.
	On the other hand, the total ${\delta_k^*}$ of some tasks on $P_x$
	is no larger than $1$, which guarantees that we can 
	divide tasks on $P_x$ to reduce its total $\delta_k$ to $1$.
	The function \Scrape($P_x$) divides container tasks on $P_x$
	and make the total load of $P_x$ to be exactly $1$ and returns the newly generated container tasks.
	The pseudo-code of \Scrape($P_x$) is shown in Algorithm \ref{al:scrape}.
	
	\item  Finally, \Partition($S$, $\Omega$) partitions all the generated container tasks in step 2) to the processors remained in $\Omega$ \\
	using the Worst-Fit packing strategy.
	After the first step, the total load of tasks on processors remained in 
	$\Omega$ is still smaller than $1$, i.e., they still have remaining available capacity and potentially can accommodate more tasks.
	\Partition($S$, $\Omega$) returns \textit{true} if tasks in $S$ can be successfully partitioned to  processors remained in $\Omega$, and returns \textit{false} otherwise.
\end{enumerate}

\begin{algorithm}
	\caption{\Sched$^*$($S$, $\Omega$) in \secondSFS.}
	\begin{algorithmic}[1]
		\STATE Sort elements in $S$ in non-increasing order of their $\delta_i^*$
		\STATE $\Psi = \emptyset$
		\FOR{each sequential task $\varphi_k$ (including both container tasks and light tasks)} \label{line:second-3}
		\STATE $P_x =$ a processor in $\Omega$ with the minimal $ \sum_{\varphi_i \in \varphi(P_x)} \delta_i^*$ and satisfying 
		\vspace{-0.12in}
		\[
		\delta_k^* + \sum_{\varphi_i \in \varphi(P_x)} \delta_i^*  <=1
		\vspace{-0.1in} 	
		\]
		\label{line:second-4}
		\STATE \textbf{if} {$P_x = \textrm{NULL}$}
		\textbf{then} \textbf{return} \textit{failure};
		\label{line:second-5}
		\STATE $\varphi(P_x) = \varphi(P_x) \cup \{\varphi_k\}$   		\label{line:second-6}
		\STATE \textbf{if} {$\sum_{\varphi_i \in \varphi(P_x)} \delta_i > 1$}
		\textbf{then} move $P_x$ from $\Omega$ to $\Psi$
		\label{line:second-7}
		\ENDFOR \label{line:second-8}
		\STATE $S = \emptyset$
		\FOR{each core $P_x \in \Psi$ }
		\STATE $S = S ~ \cup $ \Scrape($P_x$); 
		\ENDFOR
		\STATE 
		\textbf{if} \Partition($S$, $\Omega$) \textbf{then} \textbf{return} \textit{success} \textbf{else return} \textit{failure}
	\end{algorithmic}
	\label{al:second}
\end{algorithm}

\begin{algorithm}
	\caption{\Scrape($P_x$).}
	\begin{algorithmic}[1]
		\STATE{ $SS = \emptyset$}
		\STATE{ $w = \sum_{\varphi_k \in \varphi(P_x)} \delta_k - 1$}
		\FOR{each container task $\varphi_k \in \varphi(P_x)$} \label{line:second-1}
		\IF{$\delta_k - \delta_k^* > w$}
		\STATE divide $\varphi_k$ into $\varphi_k'$ and $\varphi_k''$ such that
		\[
		\delta_k'' =  w \wedge  
		\delta_k' = \delta_k - \delta_k''
		\]
		\STATE put $\varphi_k''$ in $SS$ ($\varphi_k'$ still assigned to $P_x$)
		\STATE \textbf{return} $SS$
		\ELSE
		\STATE divide $\varphi_k$ into $\varphi_k'$ and $\varphi_k''$ such that
		\[
		\delta_k' =  \delta_k^* \wedge  \delta_k'' = \delta_k - \delta_k^*
		\]
		\STATE put $\varphi_k''$ in $SS$ ($\varphi_k'$ still assigned to $P_x$)	
		\STATE $w = w -  \delta_k''$					
		\ENDIF		
		\ENDFOR
	\end{algorithmic}
	\label{al:scrape}
\end{algorithm}

Appendix-C includes an example to illustrate \secondSFS. 

The number of extra vertices created by runtime dispatching of each DAG task in \secondSFS\ is bounded as follows.

\begin{theorem}\label{t:secondoverhead}
Under \secondSFS,
the number of extra vertices created in each DAG task is bounded by $2N$, where $N$ is the number of vertices in the original DAG. 
\end{theorem}
The intuition of the proof is similar to that of Theorem \ref{t:firstoverhead}. The difference is that
\secondSFS\ uses two container tasks, so the number of splitting is doubled in the worst-case. A complete proof of the theorem is provided in Appendix-D.

\section{Performance evaluations}

\begin{figure*}
	\centering
c	\subfigure[$m=8$]{\includegraphics[width=2in]{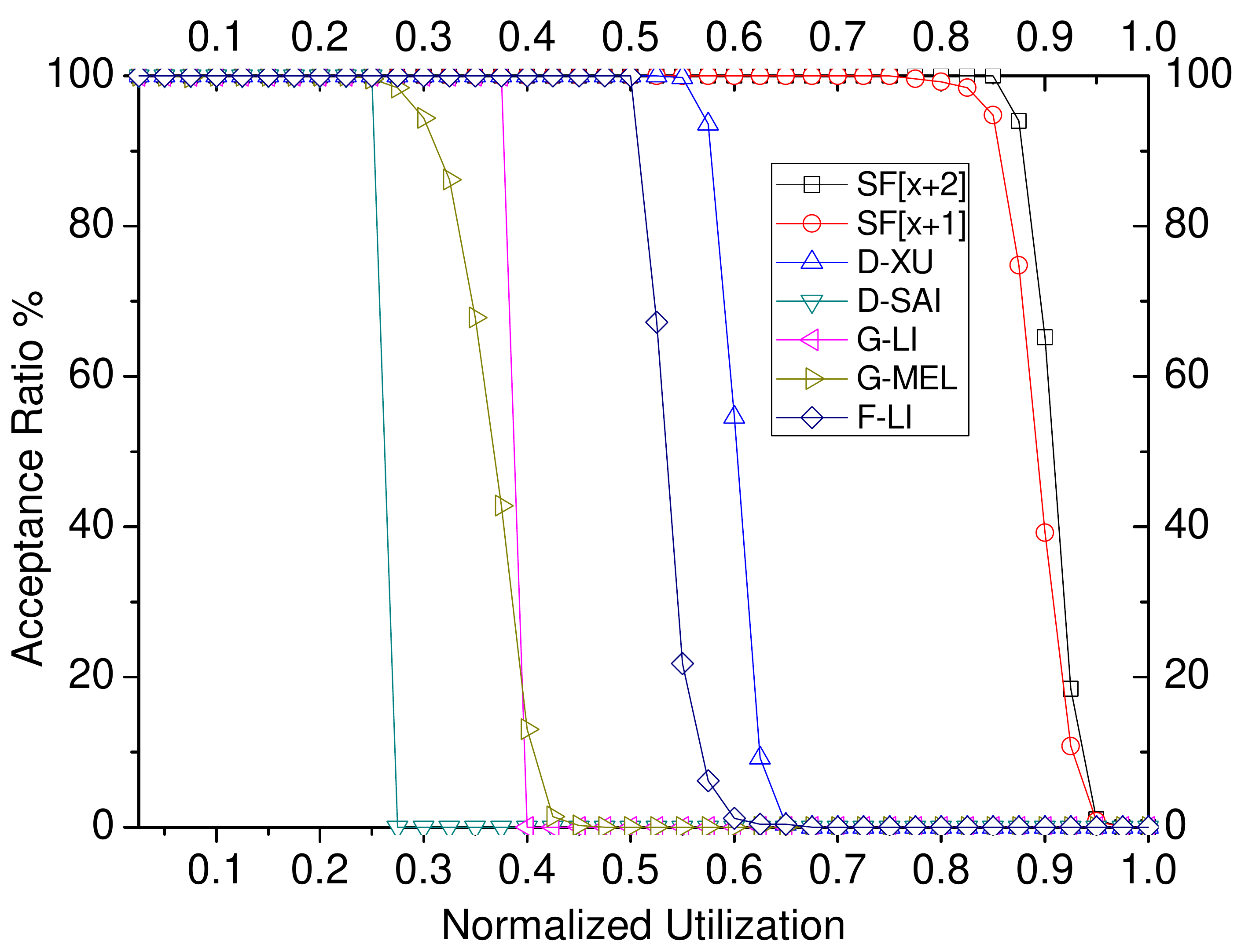}}
	\hspace{0.2in}
	\subfigure[$m=16$]{\includegraphics[width=2in]{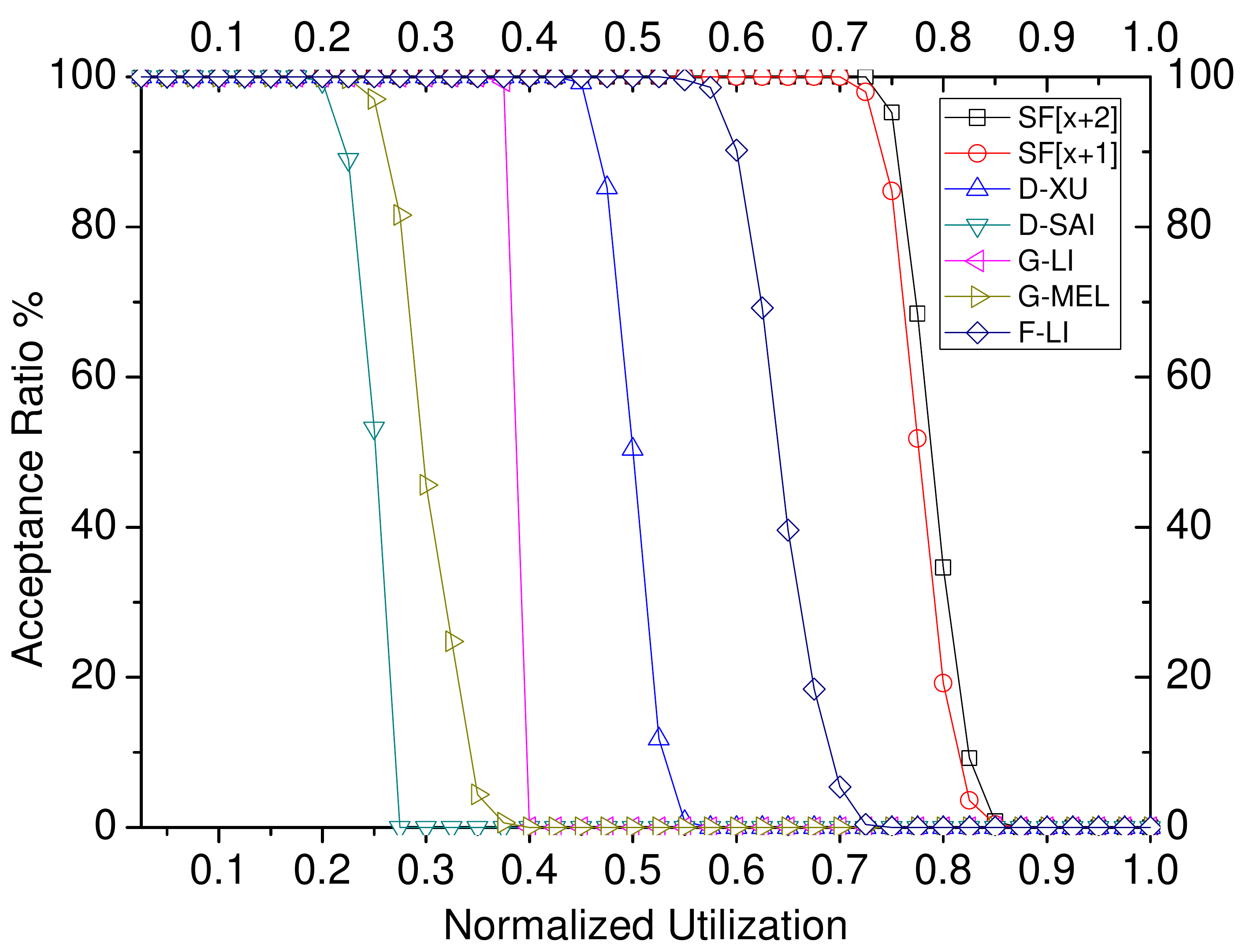}}
	\hspace{0.2in}
	\subfigure[$m=32$]{\includegraphics[width=2in]{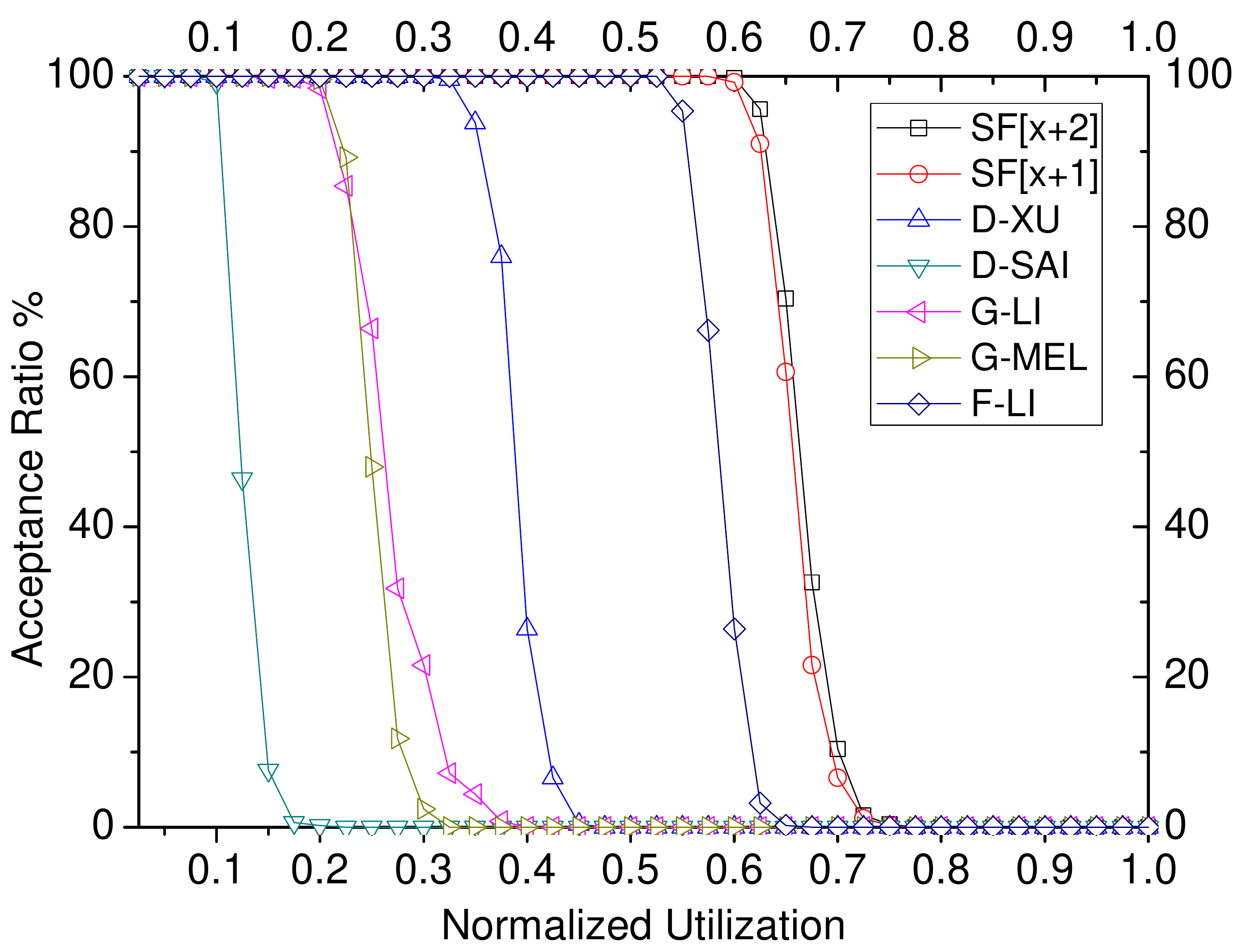}}
	\caption{Comparing \firstSFS\ and \secondSFS\ with the state-of-the-art with different number of processors.}
	\label{fig:8cores}
\end{figure*}

\begin{figure*}
	\centering
	\subfigure[Comparison with different $p$.]{\includegraphics[width=2in]{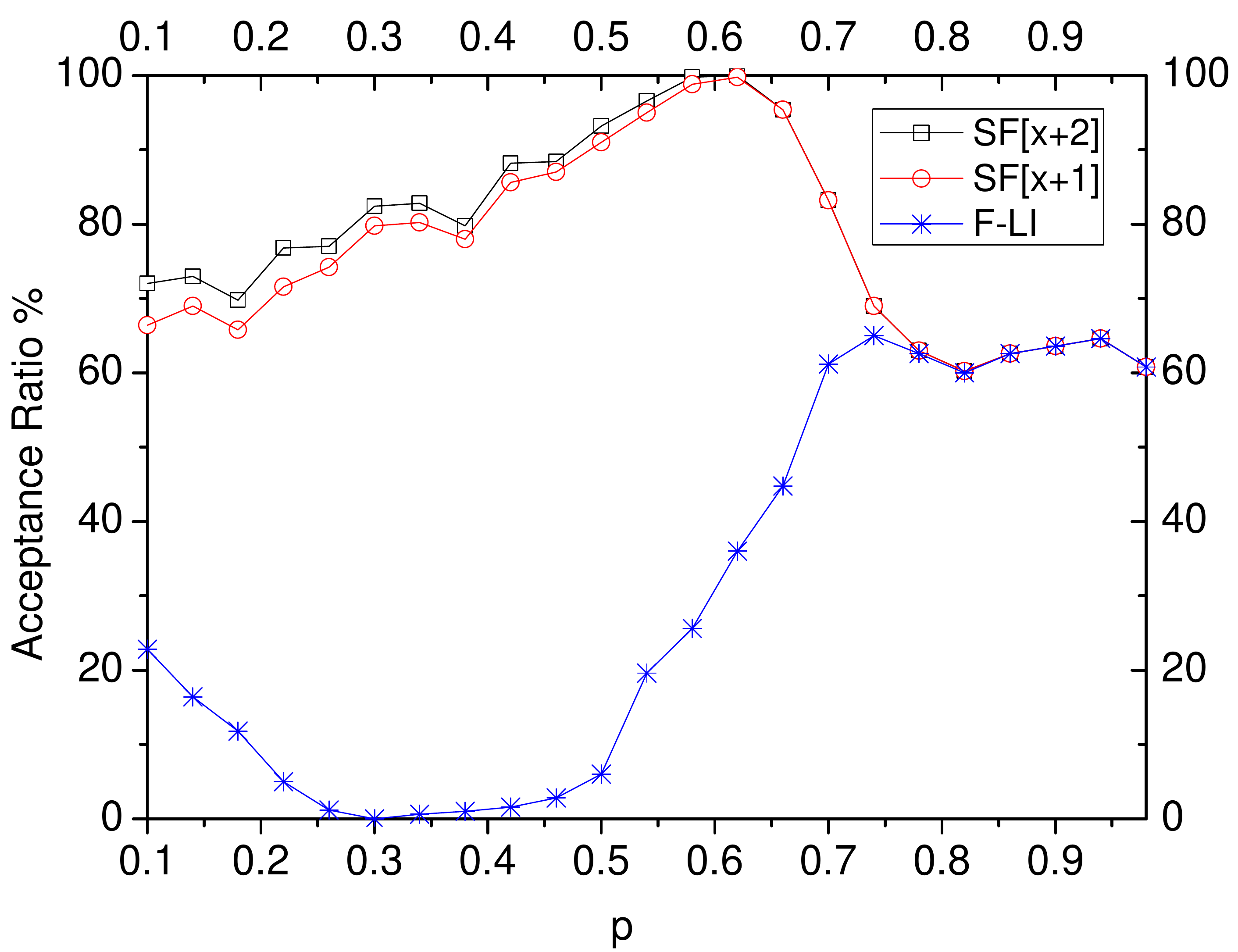}}
		\hspace{0.2in}
	\subfigure[Comparison with different average $\gamma_i$.]{\includegraphics[width=2in]{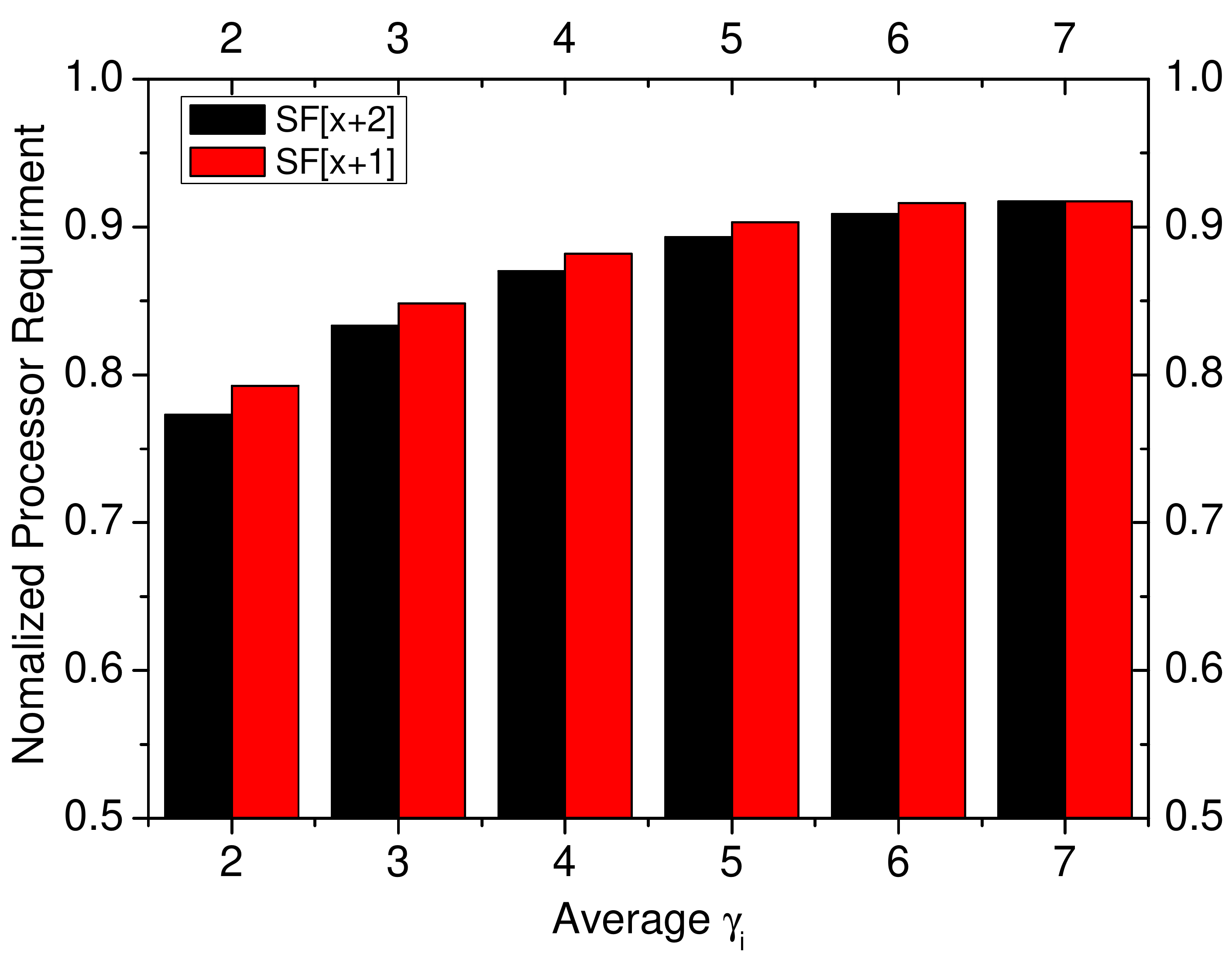}}	
		\hspace{0.2in}
		\subfigure[Comparison with different average seqeutial task load.]{\includegraphics[width=2in]{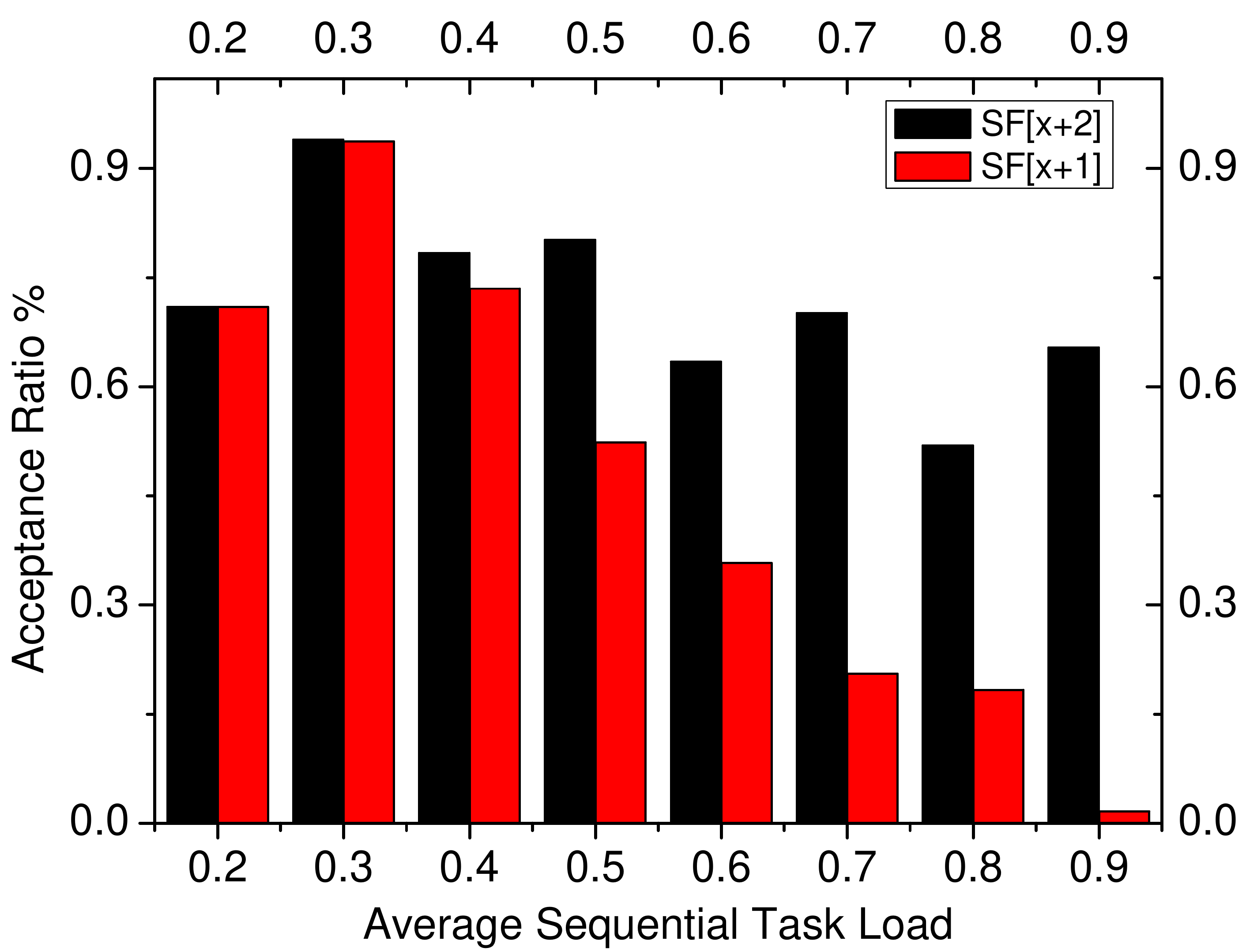}}
	\caption{Comparing \firstSFS\ and \secondSFS\ with federated scheduling in different dimensions.}
	\label{fig:compare-2}
\end{figure*}
In this section, we evaluate the performance of the proposed 
semi-federated algorithms. 
We compare the
acceptance ratio of \firstSFS\ and \secondSFS with 
the state-of-the-art algorithms and analysis methods in all the three types of parallel real-time task scheduling algorithms: 
\begin{itemize}
	\item \textbf{Decomposition-based scheduling}: 
(i) The EDF-based scheduling and analysis techniques developed in \cite{saifullah2014parallel}, denoted by $\DSAI$.
(ii) The EDF-based scheduling and analysis techniques in \cite{jiang2016decomposition}, denoted by $\DXU$.

	\item \textbf{Global scheduling}: (i) The schedulability test based on capacity augmentation bounds for global EDF scheduling in \cite{li2014analysis}, denoted by $\GLI$. (ii) The schedulability test based on response time analysis for global EDF scheduling in \cite{melani2015response}, denoted by $\GMEL$. $\GMEL$ was developed for a more general DAG model with conditional branching, but can be directly
		applied to the DAG model of this paper, which is a special case of \cite{melani2015response}.

	\item \textbf{Federated scheduling}: the schedulability test based on the processor allocation strategy in \cite{li2014analysis}, denoted by $\FLI$.
\end{itemize}
Other methods not included in our comparison are either theoretically dominated or shown to be significantly outperformed (with empirical evaluations) by one of the above  methods.

The task sets are generated using the Erd\"{o}s-R\'{e}nyi method $G(n_i , p)$ \cite{cordeiro2010random}.
For each task, the number of vertices is randomly chosen in the range $[50, 250]$ and the worst-case execution time of each vertex
is randomly picked in the range $[50, 100]$, and a valid period is generated according to  a similar method with \cite{saifullah2014parallel}. The period $T_i$ (we set $D_i=T_i$) is set to be $(L_i+\frac{C_i}{0.4m\ast{U}})*(1+0.25*Gamma(2,1))$ where $Gamma(2,1)$ is a random value by using gamma distribution and $U$ is the normalized utilization of the task set (total utilization divided by the number of processors $m$). In this way, we can: (i) make a valid period, (ii) generate a reasonable number of tasks when the processor number and total utilization of the task sets change.
For each possible edge we generate a random value in the range $[0, 1]$
and add the edge to the graph only if the generated value is less than a predefined threshold $p$.
In general the critical path of a DAG generated using the Erd\"{o}s-R\'{e}nyi method becomes longer as $p$ increases, which makes the task more sequential.
We compare the \emph{acceptance ratio} of each method, which is the ratio between the number of task sets deemed to be schedulable by a method and the total number of task sets in the experiments (of a specific
group). For each parameter configuration, we generate 10000 task sets.

Figure \ref{fig:8cores} compares the acceptance ratios of different methods 
with fixed $p=0.1$ and different number of processors. In each figure, the experiment results are grouped by the normalized utilization task set (x-axis).
We can see that our two semi-federated scheduling algorithms significantly outperform all the state-of-the-art methods
in different categories. 

Then we made in-depth comparison between federated scheduling (\FLI) and our semi-federated scheduling algorithms.
Figure \ref{fig:compare-2}-(a) shows 
the acceptance ratio with $m=16$ and different $p$ values (x-axis). We can see that semi-federated scheduling significantly 
outperforms federated scheduling except when 
$p$ is large, i.e., when tasks are very sequential. In the extreme case, when tasks 
are all sequential, both federated and semi-federated scheduling degrade to 
traditional multiprocessor scheduling of sequential tasks.

Figure \ref{fig:compare-2}-(b) compares the minimal number of processors required by the federated scheduling and semi-federated scheduling algorithms to make the task set schedulable.
In these experiments we set $p=0.1$.
The experiment results are grouped by the average minimal capacity requirement $\gamma_i$ of all heavy tasks in a task set.
A value $x$ on the x-axis represents  
range $(x-1, x]$.
The y-axis is the average ratio between the minimal number of processors required by \firstSFS (\secondSFS) 
and the minimal number of processors required by 
\FLI, to make the task set schedulable. We can see the resource saving by 
\firstSFS (\secondSFS) is more significant when $\gamma_i$ is smaller. 

Figure \ref{fig:compare-2}-(c) compares our two semi-federated scheduling algorithms, 
in which all task sets have a fixed total normalized utilization $0.8$, and we set $m=16$
and $p=0.1$.
The experiment results are grouped 
by the average load of the sequential tasks
(container tasks with fractional load bounds and light tasks) participating the partitioning on the shared processors (i.e., tasks in $S$ for
\Sched($S$, $\Omega$) and \Sched$^*$($S$, $\Omega$)).
A value $x$ on the x-axis represents  
range $(x-0.1, x]$.
As expected, when the task sizes are larger, the 
performance of \firstSFS\ degrades.
\secondSFS\ maintains good performance with large tasks since dividing a large container task into two significantly improves resource utilization.

\section{Related Work}\label{s:related}
 Early work on real-time scheduling of parallel tasks assume restricted constraints on task structures \cite{manimaran1998new,lee2006optimal,kato2009gang,lakshmanan2010scheduling,saifullah2013multi,kim2013parallel,nelissen2012techniques,maia2014response,andersson2012analyzing, axer2013response}.
For example, a Gang EDF scheduling algorithm was proposed in \cite{kato2009gang} for moldable parallel tasks. The parallel synchronous task model was studied in \cite{lakshmanan2010scheduling, saifullah2013multi,kim2013parallel,nelissen2012techniques,maia2014response,andersson2012analyzing, axer2013response}.
Real-time scheduling algorithms for DAG tasks can be classified into three paradigms:
(i) decomposition-based scheduling \cite{saifullah2014parallel,qamhieh2013global,qamhieh2014stretching,jiang2016decomposition}, (ii) global scheduling (without decomposition) \cite{bonifaci2013feasibility,baruah2014improved,melani2015response} ,
and (iii) federated scheduling \cite{li2014analysis,baruah2015federatedconstraind,baruah2015federated,baruah2015federatedconditional} .

The decomposition-based scheduling algorithms transform each DAG into several sequential sporadic sub-tasks and schedule them by traditional multiprocessor scheduling algorithms.
In \cite{saifullah2014parallel}, a capacity augmentation bound of $4$ was proved for global EDF. A schedulability test in \cite{qamhieh2013global} was provided to achieve a lower capacity augmentation bound in most cases, while in other cases above $4$. In \cite{qamhieh2014stretching},  a capacity augmentation bound of $\frac{3+\sqrt{5}}{2}$ was proved for some special task sets. In \cite{jiang2016decomposition}, a decomposition strategy exploring the structure features of the DAG was proposed, which has capacity augmentation bound between $2$ and $4$, depending on the DAG structure.

For global scheduling (without decomposition), a resource augmentation bound of $2$ was proved in \cite{baruah2012generalized} for a single DAG. In \cite{bonifaci2013feasibility}, \cite{li2015global}, a resource augmentation bound of $2-1/m$
and a capacity augmentation bound of $4-2/m$ were proved under global EDF.
A pseudo-polynomial time sufficient schedulability test was presented in \cite{bonifaci2013feasibility},
which later was generalized and dominated by \cite{baruah2014improved} for constrained deadline DAGs.
\cite{li2015global} proved the capacity augmentation bound  $\frac{3+\sqrt{5}}{2}$ for EDF and 3.732 for RM. In \cite{parri2015response} a schedulability test for arbitrary deadline DAG was derived based on response-time analysis.

For federated scheduling, \cite{li2014analysis} proposed an algorithm for DAGs with implicit deadline which has a capacity augmentation bound of $2$. Later, federated scheduling was generalized to constrained-deadline DAGs\cite{baruah2015federatedconstraind}, arbitary-deadline DAGs \cite{baruah2015federated} as well as DAGs
with conditional branching \cite{baruah2015federatedconditional}.  


The scheduling and analysis of 
sequential real-time tasks on
\textit{uniform} multiprocessors 
was studied in \cite{funk2001line, funk03ecrts, baruah08rtss}. Recently, Yang and Anderson \cite{yang2014optimal} investigated global EDF scheduling of
npc-sporadic (no precedence constraints) tasks on uniform multiprocessor platform. This study was later extended to DAG-based task model on heterogeneous multiprocessors platform in \cite{yang2016reducing} where a release-enforcer technique was used to transformed a DAG-based task into several npc-sporadic jobs thus eliminating the intra precedence constraints and provide analysis upon the response time.

\section{CONCLUSIONS}

We propose the
semi-federate scheduling approach to solve the resource waste problem
of federated scheduling.
Experimental results show significantly performance improvements of our approach comparing with the state-of-the-art for scheduling parallel real-time tasks on multi-cores. 
In the next step, we will integrate our approach with the work-stealing strategy \cite{li2016randomized} to support hight resource utilization with both hard real-time and soft real-time tasks at the same time.

\bibliographystyle{IEEEtran}
\bibliography{IEEEabrv,ADA1}

\begin{thebibliography}{10}
\providecommand{\url}[1]{#1}
\csname url@samestyle\endcsname
\providecommand{\newblock}{\relax}
\providecommand{\bibinfo}[2]{#2}
\providecommand{\BIBentrySTDinterwordspacing}{\spaceskip=0pt\relax}
\providecommand{\BIBentryALTinterwordstretchfactor}{4}
\providecommand{\BIBentryALTinterwordspacing}{\spaceskip=\fontdimen2\font plus
\BIBentryALTinterwordstretchfactor\fontdimen3\font minus
  \fontdimen4\font\relax}
\providecommand{\BIBforeignlanguage}[2]{{%
\expandafter\ifx\csname l@#1\endcsname\relax
\typeout{** WARNING: IEEEtran.bst: No hyphenation pattern has been}%
\typeout{** loaded for the language `#1'. Using the pattern for}%
\typeout{** the default language instead.}%
\else
\language=\csname l@#1\endcsname
\fi
#2}}
\providecommand{\BIBdecl}{\relax}
\BIBdecl

\bibitem{li2014analysis}
J.~Li, J.~J. Chen, K.~Agrawal, C.~Lu, C.~Gill, and A.~Saifullah, ``Analysis of
  federated and global scheduling for parallel real-time tasks,'' in
  \emph{ECRTS}, 2014.

\bibitem{baruah2015global}
S.~Baruah, V.~Bonifaci, and A.~Marchetti-Spaccamela, ``The global edf
  scheduling of systems of conditional sporadic dag tasks,'' in \emph{ECRTS},
  2015.

\bibitem{melani2015response}
A.~Melani, M.~Bertogna, V.~Bonifaci, A.~Marchetti-Spaccamela, and G.~C.
  Buttazzo, ``Response-time analysis of conditional dag tasks in multiprocessor
  systems,'' in \emph{ECRTS}, 2015.

\bibitem{graham1969bounds}
R.~L. Graham, ``Bounds on multiprocessing timing anomalies,'' \emph{SIAM
  journal on Applied Mathematics}, 1969.

\bibitem{baruah07rtss}
S.~Baruah, ``Techniques for multiprocessor global schedulability analysis,''
  \emph{RTSS}, 2007.

\bibitem{baruah05rtss}
S.~Baruah and N.~Fisher, ``The partitioned multiprocessor scheduling of
  sporadic task systems,'' \emph{RTSS}, 2005.

\bibitem{funk2001line}
S.~Funk, J.~Goossens, and S.~Baruah, ``On-line scheduling on uniform
  multiprocessors,'' in \emph{RTSS}, 2001.

\bibitem{johnson1974worst}
D.~S. Johnson, A.~Demers, J.~D. Ullman, M.~R. Garey, and R.~L. Graham,
  ``Worst-case performance bounds for simple one-dimensional packing
  algorithms,'' \emph{SIAM Journal on Computing}, 1974.

\bibitem{baruah99gmf}
S.~Baruah, D.~Chen, S.~Gorinsky, and A.~Mok, ``Generalized multiframe tasks,''
  \emph{Real-Time Systems}, 1999.

\bibitem{saifullah2014parallel}
A.~Saifullah, D.~Ferry, J.~Li, K.~Agrawal, C.~Lu, and C.~D. Gill, ``Parallel
  real-time scheduling of dags,'' \emph{Parallel and Distributed Systems, IEEE
  Transactions on}, 2014.

\bibitem{jiang2016decomposition}
X.~Jiang, X.~Long, N.~Guan, and H.~Wan, ``On the decomposition-based global edf
  scheduling of parallel real-time tasks,'' in \emph{RTSS}, 2016.

\bibitem{cordeiro2010random}
D.~Cordeiro, G.~Mouni{\'e}, S.~Perarnau, D.~Trystram, J.-M. Vincent, and
  F.~Wagner, ``Random graph generation for scheduling simulations,'' in
  \emph{ICST}, 2010.

\bibitem{manimaran1998new}
G.~Manimaran, C.~S.~R. Murthy, and K.~Ramamritham, ``A new approach for
  scheduling of parallelizable tasks in real-time multiprocessor systems,''
  \emph{Real-Time Systems}, 1998.

\bibitem{lee2006optimal}
W.~Y. Lee and L.~Heejo, ``Optimal scheduling for real-time parallel tasks,''
  \emph{IEICE transactions on information and systems}, 2006.

\bibitem{kato2009gang}
S.~Kato and Y.~Ishikawa, ``Gang edf scheduling of parallel task systems,'' in
  \emph{RTSS}, 2009.

\bibitem{lakshmanan2010scheduling}
K.~Lakshmanan, S.~Kato, and R.~Rajkumar, ``Scheduling parallel real-time tasks
  on multi-core processors,'' in \emph{RTSS}, 2010.

\bibitem{saifullah2013multi}
A.~Saifullah, J.~Li, K.~Agrawal, C.~Lu, and C.~Gill, ``Multi-core real-time
  scheduling for generalized parallel task models,'' \emph{Real-Time Systems},
  2013.

\bibitem{kim2013parallel}
J.~Kim, H.~Kim, K.~Lakshmanan, and R.~R. Rajkumar, ``Parallel scheduling for
  cyber-physical systems: Analysis and case study on a self-driving car,'' in
  \emph{ICCPS}, 2013.

\bibitem{nelissen2012techniques}
G.~Nelissen, V.~Berten, J.~Goossens, and D.~Milojevic, ``Techniques optimizing
  the number of processors to schedule multi-threaded tasks,'' in \emph{ECRTS},
  2012.

\bibitem{maia2014response}
C.~Maia, M.~Bertogna, L.~Nogueira, and L.~M. Pinho, ``Response-time analysis of
  synchronous parallel tasks in multiprocessor systems,'' in \emph{RTNS}, 2014.

\bibitem{andersson2012analyzing}
B.~Andersson and D.~de~Niz, ``Analyzing global-edf for multiprocessor
  scheduling of parallel tasks,'' in \emph{OPODIS}, 2012.

\bibitem{axer2013response}
P.~Axer, S.~Quinton, M.~Neukirchner, R.~Ernst, B.~Dobel, and H.~Hartig,
  ``Response-time analysis of parallel fork-join workloads with real-time
  constraints,'' in \emph{ECRTS}, 2013.

\bibitem{qamhieh2013global}
M.~Qamhieh, F.~Fauberteau, L.~George, and S.~Midonnet, ``Global edf scheduling
  of directed acyclic graphs on multiprocessor systems,'' in \emph{RTNS}, 2013.

\bibitem{qamhieh2014stretching}
M.~Qamhieh, L.~George, and S.~Midonnet, ``A stretching algorithm for parallel
  real-time dag tasks on multiprocessor systems,'' in \emph{RTNS}, 2014.

\bibitem{bonifaci2013feasibility}
V.~Bonifaci, A.~Marchetti-Spaccamela, S.~Stiller, and A.~Wiese, ``Feasibility
  analysis in the sporadic dag task model,'' in \emph{ECRTS}, 2013.

\bibitem{baruah2014improved}
S.~Baruah, ``Improved multiprocessor global schedulability analysis of sporadic
  dag task systems,'' in \emph{ECRTS}, 2014.

\bibitem{baruah2015federatedconstraind}
------, ``The federated scheduling of constrained-deadline sporadic dag task
  systems,'' in \emph{DATE}, 2015.

\bibitem{baruah2015federated}
------, ``Federated scheduling of sporadic dag task systems,'' in \emph{IPDPS},
  2015.

\bibitem{baruah2015federatedconditional}
------, ``The federated scheduling of systems of conditional sporadic dag
  tasks,'' in \emph{EMSOFT}, 2015.

\bibitem{baruah2012generalized}
S.~Baruah, V.~Bonifaci, A.~Marchetti-Spaccamela, L.~Stougie, and A.~Wiese, ``A
  generalized parallel task model for recurrent real-time processes,'' in
  \emph{RTSS}, 2012.

\bibitem{li2015global}
J.~Li, K.~Agrawal, C.~Lu, and C.~Gill, ``Outstanding paper award: Analysis of
  global edf for parallel tasks,'' in \emph{ECRTS}, 2013.

\bibitem{parri2015response}
A.~Parri, A.~Biondi, and M.~Marinoni, ``Response time analysis for g-edf and
  g-dm scheduling of sporadic dag-tasks with arbitrary deadline,'' in
  \emph{RTNS}, 2015.

\bibitem{funk03ecrts}
S.~Funk and S.~Baruah, ``Characteristics of edf schedulability on uniform
  multiprocessors,'' \emph{ECRTS}, 2003.

\bibitem{baruah08rtss}
S.~Baruah and J.~Goossens:, ``The edf scheduling of sporadic task systems on
  uniform multiprocessors,'' \emph{RTSS}, 2008.

\bibitem{yang2014optimal}
K.~Yang and J.~H. Anderson, ``Optimal gedf-based schedulers that allow
  intra-task parallelism on heterogeneous multiprocessors,'' in
  \emph{ESTIMedia}, 2014.

\bibitem{yang2016reducing}
K.~Yang, M.~Yang, and J.~H. Anderson, ``Reducing response-time bounds for
  dag-based task systems on heterogeneous multicore platforms,'' in
  \emph{RTNS}, 2016.

\bibitem{li2016randomized}
J.~Li, S.~Dinh, K.~Kieselbach, K.~Agrawal, C.~Gill, and C.~Lu, ``Randomized
  work stealing for large scale soft real-time systems,'' in \emph{RTSS}, 2016.

\end{thebibliography}

\newpage
\section*{Appendix-A: Response Time Bounds without  Inter-Processor Migration}

The work-conserving scheduling rules for uniform multiprocessors in Section \ref{ss:work-conserving}
requires the vertices to migrate from slower processors to faster processors whenever possible.
If such migration is forbidden, the resource may be significantly wasted and the response time can be much larger.
We say a scheduling algorithm is \textit{weakly work-conserving} if only the first work-conserving rule in Section \ref{ss:work-conserving}
is satisfied and a vertex is not allowed to migrate from one processor to another. The response time of a DAG task under weakly conserving scheduling is bounded by the following theorem:

\begin{theorem}\label{l:withoutmigration}
	Give $m$ uniform processors with speeds $\{\speed_1,\speed_2,...,\speed_m\}$ (sorted in non-increasing order).  
	The response time of a DAG task $\tau_i$ by a weakly work-conserving scheduling algorithm is bounded by
	\begin{equation}\label{e:boundwithoutmigration}
	\response\leq\frac{L_i}{\speed_m}+\frac{C_i-L_i}{S_m}
	\end{equation}
\end{theorem}

\begin{proof}	
	Without loss of generality, we assume the task under analysis releases an
	instance at time 0, and thus $R$ is the time point when the
	currently release of $\tau_i$ is finished. In the time window $[0, R]$,
	let $\alpha$ denote the total length of intervals during which the at least one processor is idle and $\beta$ denote the total length of the intervals during which all processors are busy.
	Therefore, we know $R = \alpha + \beta$. 
Let $v_z$ be the latest finished vertex in the DAG and $v_{z-1}$ be the latest finished vertex among all predecessors of $v_z$, repeat this way until no predecessor can be found, we can simply construct a path $\pi=\{v_1,v_2,...,v_{z-1},v_z\}$. The fundamental observation is that all processors must be busy between the finishing time of $v_k$ and the starting time of $v_{k+1}$ where $1\leq{k}\leq{z-1}$.	\forget{Let $\platform$ be an arbitrary path in the DAG starting from the head vertex and ending at the tail vertex.} We use $\consume(\platform)$ to denote the total amount of workload executed for vertices along path $\platform$ in all the time intervals of $\alpha$. The total work done in all time interval of $\beta$ is at most $C_i-\consume(\platform)$. Since at least
	one processor is idle in time intervals of $\alpha$, $\platform$ must contain a vertex being executed
	in these time intervals (since at any time point before $R$, there
	is at least one eligible vertex along any path) and $\speed_m$ is the speed of the slowest processor. Therefore, we know:
	\[
	\alpha\leq\frac{\consume(\platform)}{\speed_m}
	\]
	As the total work been done in all time interval of $\beta$ (where all processors are busy) is at most $C_i-\consume(\platform)$, we have 
	\begin{align*}
		\beta\sumspeed_m&\leq{C_i-\consume(\platform)}\\
	\Leftrightarrow~~~~~~	\beta&\leq\frac{C-\consume(\platform)}{\sumspeed_m}
	\end{align*}
	Hence we have 
	\begin{equation}\label{e:append-1}
	R=\alpha+\beta\leq \frac{\consume(\platform)}{\speed_m} + \frac{C_i-\consume(\platform)}{\sumspeed_m}
	\end{equation}
	Let $l_{\platform}$ denote the total workload (of all vertices) along path
	$\platform$, so we know $\consume(\platform)\leq{l_{\platform}}$.
	Since $L_i$ is the total workload of the longest path in the DAG,
	we know $l_{\platform}\leq{L_i}$, in summary we have $\consume(\platform) \leq L_i$
	and applying this to (\ref{e:append-1}) concludes the lemma.
\end{proof}

	 \begin{figure}[!htb]
	 	\centering
	 	\subfigure[A DAG task example.]{\includegraphics[width=1.6in]{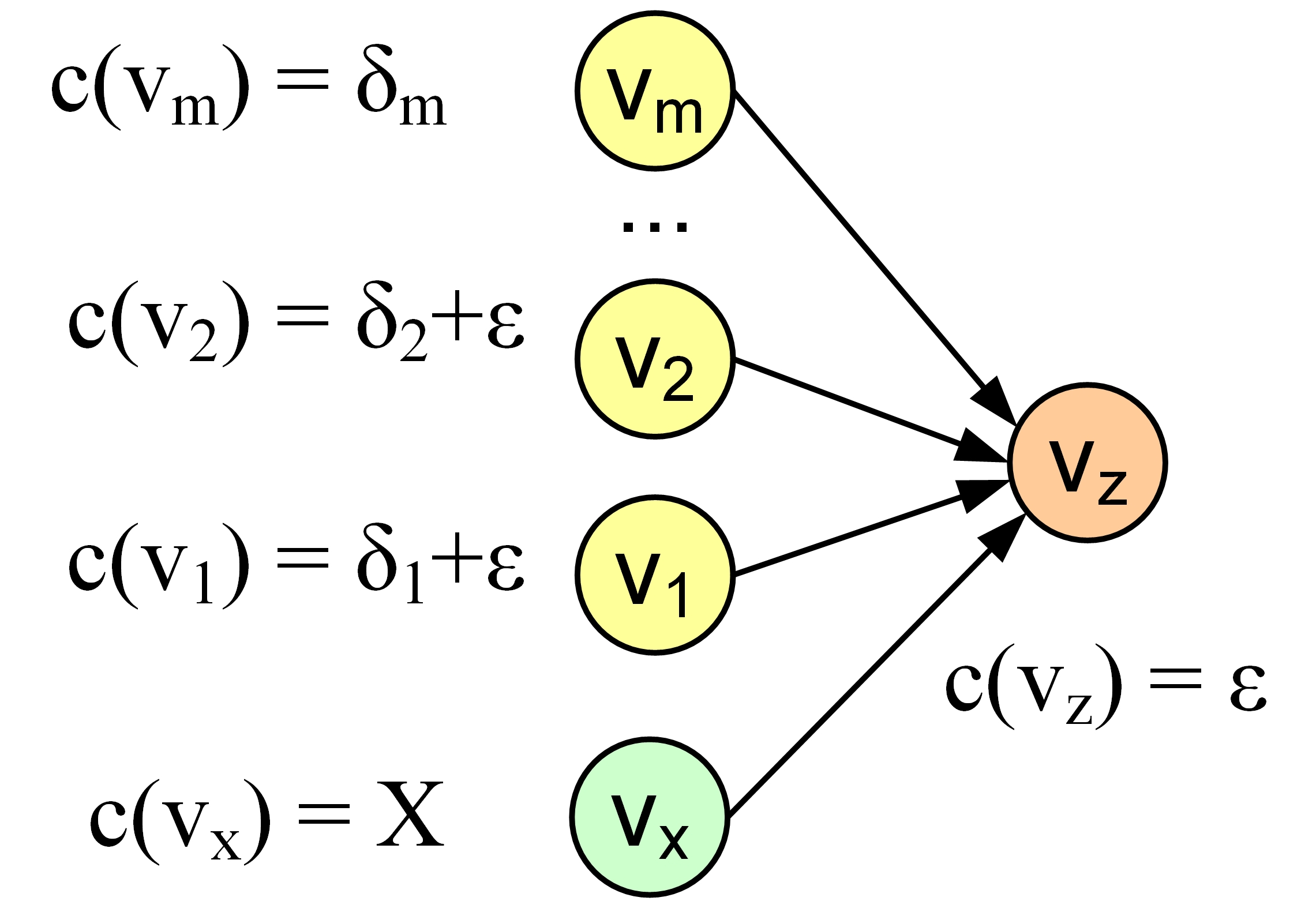}}
	 	\subfigure[Scheduling sequence.]{\includegraphics[width=1.7in]{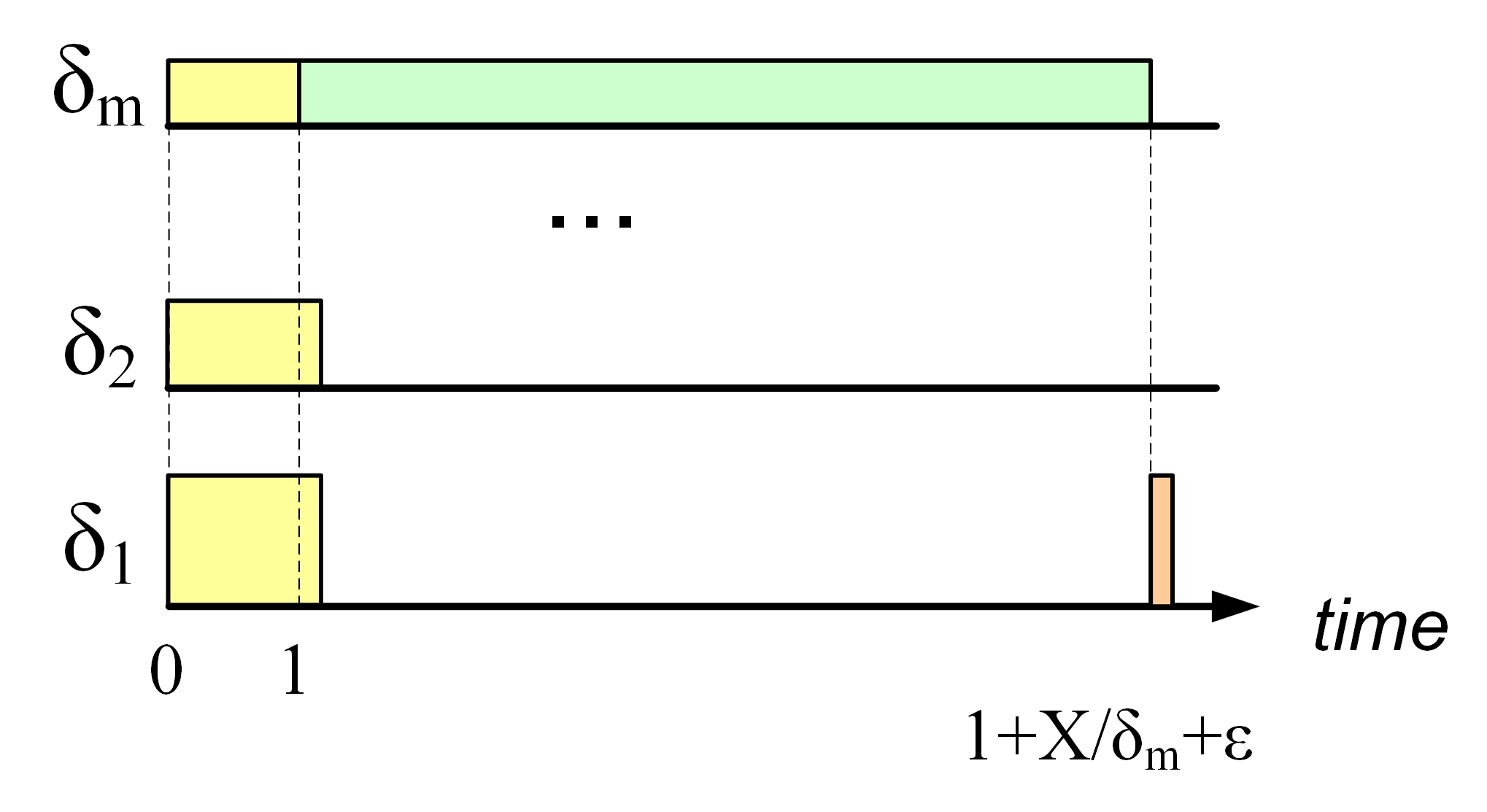}}
	 	\caption{Scheduling of a DAG task without migration.}
	 	\label{fig:tight_example}
	 \end{figure}

\begin{theorem}\label{l:tight}
The response time bound in (\ref{e:boundwithoutmigration})
is tight for weakly work-conserving scheduling algorithms.
\end{theorem}
\begin{proof}
The tightness is witnessed by the example in Figure 	\ref{fig:tight_example}. Let $R'$ denote its actual response time of the task in the scheduling sequence in Figure 
\ref{fig:tight_example}-(b), and $R''$ the response time bound in \ref{e:boundwithoutmigration}.
$R'/R''$ approaches $1$ as $X$ approaches infinity.
\end{proof}
As illustrated in Figure \ref{fig:tight_example}, if inter-processor migration is forbidden, the response time is almost the same as executing the critical path on the
slowest processor, while the faster processors
are all wasted.

\section*{Appendix-B: Detailed Explanation of The Example in Figure \ref{fig:deadlineassignment} }

\begin{itemize}
	\item At time $t = 0$, the only eligible vertex $v_1$ starts to occupy the \enquote{fastest} container task $\varphi_1$,
	and its absolute deadline is set to be $d_1 = 0 + 1/1 = 1$.
	
	\item At time $t=1$, $\varphi_1$ becomes empty, and $v_2$, $v_3$ and $v_4$ become eligible. Suppose we first select $v_4$ to execute on $\varphi_1$, with $d_1 = 1 + 4/1 = 5$. After that, Algorithm 
	\ref{a:dispatcher} is invoked again to assign a container task to the next eligible vertex $v_3$.
	If we encapsulate the entire $v_3$ into $\varphi_2$, then the resulting absolute deadline
	$1 + 3/0.5 = 7$ is later than the absolute deadline of a \enquote{faster} container task $\varphi_1$'s deadline $d_1 = 5$. Therefore, we must split $v_3$ into $v_3'$ and $v_3''$, so that putting $v_3'$ into $\varphi_2$ results in the same absolute deadline as
	$\varphi_1$, and $v_3''$ is put back to $S$ for further consideration. Next, Algorithm 
	\ref{a:dispatcher} is invoked again to assign the only eligible vertex $v_2$ to the remaining container task $\varphi_3$. Similarly, $v_2$ cannot be put into 
	$\varphi_3$ entirely, and we split it into $v_2'$ and $v_2''$ so that $d_3 = 5$.
	
	\item At time $t=5$, all the container tasks reach their absolute deadlines and thus becomes empty, and currently only $v_2''$ and $v_3''$ are eligible. 
	Suppose we first choose to put $v_2''$ into $\varphi_1$ with $d_1 = 5 + 4/1 = 9$, then put $v_3''$ into $\varphi_2$
	with $d_2 = 5 + 1/0.5 = 7$, which is smaller than $d_1$.

	\item At time $t=7$, $\varphi_2$ reaches its absolute deadline and thus become empty, and $v_5$ become eligible and should be put into  $\varphi_2$ ($\varphi_1$ is still being occupied). $v_5$ also needs to be split into two parts $c(v_5') = 
	c(v_5'') = 1$ to make $d_2 = d_1 = 9$.
	
	\item At time $t = 9$, both $\varphi_1$ and $\varphi_2$ become empty and $v_5''$ become eligible, which is put into $\varphi_1$ with $d_1 = 10$.
	
	\item At time $10$, the execution of $v_5''$ on  $\varphi_1$ is finished and the last vertex $v_6$ is put into the fastest container task $\varphi_1$.
	
	\item At time $11$, the entire task is finished.
\end{itemize}

\section*{Appendix-C: Illustration of \firstSFS\ and \secondSFS}

We use the following example to illustrate $\firstSFS$.
Assume a task set consists of 4 DAG tasks, where
the first three are heavy, with the minimal capacity requirements $\gamma_1 = 1.6$, $\gamma_2 = 1.6$ and $\gamma_3=1.5$, and one light task with density $\gamma_4 = 0.3$.
If scheduled by standard federated scheduling, each of the three heavy tasks requires $2$ dedicated processors, and in total $7$ processor are needed. If scheduled by \firstSFS, 
each of the heavy task only requires one dedicated processors,
and they generate three container tasks, with load bounds $0.6$, $0.6$ and $0.5$. These three container tasks, together with the light tasks with density $0.3$ is schedulable by partitioned EDF on $3$ processors, so in total $6$ processors are needed to schedule the task set using \firstSFS.

We use the same task set as above to illustrate \secondSFS. Now we assume the tasks are scheduled on $5$ processors.
Since each heavy task is granted one dedicated processor, the container tasks and light task share $2$ processors.
The load bound of the three generated container tasks and the density of the light tasks are
\begin{equation*}
\delta_1 = 0.6, ~\delta_2 = 0.6, ~\delta_3 = 0.5, ~ \delta_4 = 0.3 
\end{equation*} 
We can compute 
$\delta_k^*$ for each task using (\ref{e:lowerupper}):
\begin{equation}
\delta_1^* = \frac{3}{8}, ~
\delta_2^* = \frac{3}{8}, ~
\delta_3^* = \frac{1}{3}, ~
\delta_4^* = 0.3
\end{equation} 

The algorithm \Sched$^*$($S$, $\Omega$)
works as follows:

\begin{enumerate}
	\item  $\varphi_1$ is assigned to an empty processor $P_1$. 
	
	\item  $\varphi_2$ is assigned to the other empty processor $P_2$. 

	\item To assign $\varphi_3$, both processors are holding a task with the same load, so we choose any of them, say $P_1$, to accommodate $\varphi_3$. 
	Since $\delta_1^* + \delta_3^* = 3/8 + 1/3 <1$, we can assign $\varphi_3$ to $P_1$.
	After that, since $\delta_1 + \delta_3 = 0.6 + 0.5 > 1$, $P_1$ is moved from $\Omega$ to $\Psi$.
	\item There is only one processor $P_2$ in $\Omega$, and since $\delta_2^* + \delta_4^* = 3/8 + 0.3 < 1$, we can assign  $\varphi_4$ to $P_2$. After that,
 since $\delta_2 + \delta_4 = 0.6 + 0.3 < 1$, $P_2$ remains in $\Omega$.

\item After assigning all the four tasks, only $P_1$ is in $\Psi$. So we execute $\Scrape(P_1)$.
$w = \delta_1 + \delta_3 - 1 = 0.1$.
Since $\delta_1 - \delta_1^* = 0.3 - 3/8 > 0.1$, so we 
divide $\varphi_1$ into $\varphi_1^{'}$ and $\varphi_1^{''}$ where $\speed_1^{''}=0.1$ and $\speed_1^{'}=0.6-0.1=0.5$,
and put $\varphi_1^{''}$ in $SS$.

\item There is only one processor $P_2$ in $\Omega$, 
since 
\[
\summation_{\varphi_i\in\varphi(P_1)}\speed_i+\speed_1^{''}=0.6+0.3+0.1 = 1
\]
we put $\varphi_1^{''}$ is put in $P_2$. 
\end{enumerate}
Therefore, the final result of \Sched$^*$($S$, $\Omega$) is
\begin{align*} 
\processor_1: ~&\speed_1^{'}=\frac{1}{2}, ~~\speed_3=\frac{1}{2}\\ 
\processor_2: ~&\speed_2=\frac{3}{5},~~\speed_4=\frac{3}{10},~~\speed_1^{''}=\frac{1}{10}
\end{align*}

%
%
%
%
%
%
%
%
%

\section*{Appendix-D: Proof of Theorem \ref{t:secondoverhead}}

\begin{proof}
	Let a task execute on several dedicated processors and two fractional container tasks despite the unit containers with density of $\speed^{'}$ and $\speed^{''}$, $\speed^{'} \geq \speed^{''}$.
	By the proof of Theorem \ref{t:firstoverhead} we know
	the number of splitting occurred on the container task $\delta'$ is at most $N$. In the following we prove the number of splitting on the container task $\delta''$ is also at most $N$.
	
	We use $A$ to denote the set of vertices (including the parts of the divided vertices) executed on dedicated processors, 
	and use $B$ to denote the set of vertices (parts)
	executed on container task $\delta'$ with a deadline different from any deadlines of vertices (parts)
	on the dedicated processors.
	If a vertex $v$ is divided into two parts, $v'$, 
	executed on the container task $\delta'$, and $v''$, executed on dedicated processors. The migration of $v$ must happens at a time point aligned with some deadline
	on the dedicated processors, so we know $v'$ must not be in $B$. 
	Moreover, according to Algorithm \ref{a:dispatcher}, the vertices assigned to dedicated processors will not migrate to other processors.
	Therefore, the total number
	of elements in $A \cup B $ is at most $N$. 
	Therefore, the number of time points aligned with deadlines of vertices (parts) executed on the dedicated processors and container task $\delta'$ is bounded by $N$.	
	Since a splitting on container task $\delta''$ only occurs 
	at time points aligned with deadlines of vertices (parts) executed on the dedicated processors and container task $\delta'$, we can conclude the number
	of splitting on  container task $\delta''$ is also bounded by $N$.
	 
	 In summary, the total number of vertices splitting all the two container tasks is bounded by $2N$.
	 Since the vertices assigned to dedicated processors will not migrate to other processors. Therefore, the 
	 total number of newly generated vertices is bounded by $2N$.	 
\end{proof}

\end{document}